\documentclass{interact} 
\usepackage[numbers,sort&compress]{natbib}
\bibpunct[, ]{[}{]}{,}{n}{,}{,}

\usepackage{epstopdf}
\usepackage{multirow}%
\usepackage{amsmath,amssymb,amsfonts}%
\usepackage{amsthm}%
\usepackage{thmtools}

\usepackage{mathrsfs}%
\usepackage[title]{appendix}%
\usepackage{xcolor}%
\usepackage{textcomp}%
\usepackage{manyfoot}%
\usepackage{booktabs}%
\usepackage{algorithm}%
\usepackage{algorithmicx}%
\usepackage{algpseudocode}%
\usepackage{listings}%
\usepackage{enumitem}
\usepackage{bm}
\usepackage{float}%
\usepackage{siunitx}
\usepackage{tabularx}
\usepackage{ltablex}
\usepackage{subcaption}
\usepackage{caption}
\usepackage{makecell}
\usepackage{diagbox}
\usepackage{pgfplots}
\pgfplotsset{compat=1.18}

\newcommand\setrow[1]{\gdef\rowmac{#1}#1\ignorespaces}
\newcommand\clearrow{\global\let\rowmac\relax}
\clearrow

\newtheorem{theorem}{Theorem}
\newtheorem{lemma}{Lemma}

\newtheorem{example}{Example}
\newtheorem{definition}{Definition}

\begin{document}
	\history{}
	\title{A Sequence-Form Path-Following Method for Selecting Normal-Form Perfect Equilibria via a Logistic QRE Variant}
	\author{
	\name{Yuqing Hou\textsuperscript{a,c}, Yiyin Cao\textsuperscript{b}\thanks{CONTACT Yiyin Cao. Email: yiyincao2-c@my.cityu.edu.hk}, Chuangyin Dang\textsuperscript{c}, and Yong Wang\textsuperscript{a}}
	\affil{\textsuperscript{a}Department of Automation, University of Science and Technology of China, Hefei, China; \\
	\textsuperscript{b}School of Management, Xi'an Jiaotong University, Xi'an, China; \\
	\textsuperscript{c}Department of Systems Engineering, City University of Hong Kong, Hong Kong, China}
	}
	\maketitle
	\begin{abstract}
		Logistic quantal response equilibrium (QRE)
		provides a systematic mechanism for
		equilibrium selection in extensive-form games
		through their reduced normal-form representations.
		To ensure normal-form perfection
		of the limiting equilibria,
		a perturbed logistic-QRE variant
		has been developed,
		but its direct computation
		is hindered by the exponential growth
		of the reduced normal-form strategy space. To address this difficulty, we construct an entropy-regularized artificial game in sequence form and establish that its Nash equilibria characterize the perturbed logistic-QRE variant in the reduced normal form. Based on this characterization, we develop a differentiable path-following method starting from a prescribed interior realization plan. We prove the existence of a smooth solution path and show that the mixed-strategy profiles corresponding to its limiting realization plans are normal-form perfect. Furthermore, by reformulating the dilated-entropy terms as weighted standard entropy terms, we derive an equivalent smooth system that generates the same theoretical realization-plan path. Numerical experiments illustrate the equilibrium-selection behavior of the proposed methods and demonstrate their computational advantages over the original normal-form method.
	\end{abstract}
	\par
	\noindent\textbf{Word count:} 6,393.\par
	
	\section{Introduction}
	Extensive-form games~\cite{KuhnExtensiveGames1950}
	provide a fundamental framework for modeling
	sequential strategic interactions in economics,
	political science, and engineering.
	Nash equilibrium~\cite{NashEquilibriumpointsnperson1950}
	characterizes strategy profiles from which
	no player can profitably deviate unilaterally.
	Nevertheless, an extensive-form game can have multiple Nash equilibria, some of which may deviate from our intuitive expectations regarding the game's outcome and off-equilibrium strategies. Equilibrium refinements~\cite{SeltenReexaminationperfectnessconcept1975,MyersonRefinementsNashequilibrium1978,marchiRefinementProper1995} address these limitations by imposing additional stability requirements on equilibrium behavior. Selten's normal-form perfection refines
	Nash equilibrium by requiring equilibrium
	strategies to remain optimal against
	a sequence of completely mixed strategy
	profiles converging to the equilibrium.
	This refinement excludes certain equilibria
	that rely on noncredible responses to
	small perturbations, but does not
	necessarily ensure uniqueness. Consequently, an additional criterion
	is needed to select among the remaining
	equilibria.
	
	Logistic QRE, originally introduced by McKelvey and Palfrey\cite{mckelveyQuantalResponseEquilibria1995} for normal-form games, models boundedly rational behavior through payoff-sensitive stochastic choice; when applied to an extensive-form game, it is considered on the game's associated normal-form representation. As the rationality parameter increases, the limiting Nash equilibrium reached along a continuous QRE branch provides a criterion for equilibrium selection.
	To combine the equilibrium-selection
	property of logistic QRE with the
	refinement property of normal-form perfection,
	Cao et al.~\cite{caoVariantLogisticQuantal2024} developed a perturbed
	logistic QRE variant whose limiting
	equilibria are normal-form perfect.
	However, its normal-form formulation
	suffers from exponential growth in
	the number of pure strategies,
	making direct computation expensive.
	To overcome this limitation,
	this paper develops an equivalent
	sequence-form formulation and associated
	differentiable path-following methods
	that preserve the equilibrium-selection
	and refinement properties of the variant 
	while substantially reducing computational complexity.
	
	The equilibrium-selection approach associated with logistic QRE admits a natural path-following interpretation~\cite{turocyDynamicHomotopyInterpretation2005}. More precisely, the rationality parameter indexes a continuous solution path whose limit points, as the parameter tends to infinity, are Nash equilibria. Path-following methods for computing Nash equilibria in normal-form games have been widely studied. This approach traces its origins to the Lemke–Howson complementary-pivoting algorithm for computing Nash equilibria in bimatrix games. Independent extensions by Rosenmüller~\cite{RosenmullerGeneralizationLemkeHowson1971} and Wilson~\cite{WilsonComputingEquilibriaNPerson1971} generalized the Lemke--Howson algorithm from bimatrix games to $n$-player games. Subsequent work developed simplicial methods that provided constructive and implementable procedures for computing Nash equilibria in $n$-player games~\cite{GarciaSimplicialApproximationEquilibrium1973,vanderLaanComputationFixedPoints1982,Doupnewsimplicialvariable1987,HeringsComputationNashEquilibrium2002}. However, their reliance on increasingly fine simplicial subdivisions may entail substantial computational and storage costs as the dimension of the strategy space grows, thereby limiting their scalability. These limitations motivated the development of differentiable path-following methods. Herings and Peeters~\cite{Heringsdifferentiablehomotopycompute2001} introduced the first differentiable path-following method for computing Nash equilibria, based on Harsanyi and Selten's linear tracing procedure, an early path-based procedure for Nash equilibrium selection. Govindan and Wilson~\cite{GovindanglobalNewtonmethod2003} proposed a piecewise-differentiable global Newton method for Nash equilibrium computation, while Chen and Dang~\cite{Chenreformulationbasedsmoothpathfollowing2016} later attained full differentiability through a modified logarithmic reformulation. For robust normal-form games, Liu et al.~\cite{liuDifferentiableRobust2025} used logarithmic-barrier regularization to construct a differentiable path from a totally mixed strategy profile to a Nash equilibrium. Additionally, by exploiting the equilibrium-selection properties embedded in Harsanyi's tracing procedures and logistic QRE, differentiable methods have also proven effective for the computation of a range of alternative equilibrium concepts~\cite{Chenextensionquantalresponse2020,CaovariantHarsanyitracing2022,liArbitraryStartingTracing2020,zhanDifferentiablePathfollowingAlgorithm2020}. Dreves~\cite{drevesBestResponse2019} studied equilibrium selection through a best-response method for two-player generalized Nash games.
	
	The sequence form~\cite{romanovskiiReductionGameComplete1962,Kollercomplexitytwopersonzerosum1992,vonStengelEfficientComputationBehavior1996} compactly represents extensive-form games by encoding players' strategies through action sequences and realization plans. It preserves the sequential structure of the game while scaling linearly in that of the extensive-form game, making it particularly suitable for computing equilibrium concepts defined in the normal form, whose direct treatment is hindered by the exponential growth of the normal-form strategy space. Nash equilibria are conventionally computed in normal-form games. Koller et al.~\cite{KollerEfficientComputationEquilibria1996} developed an algorithm for computing Nash equilibria of two-player extensive-form games by applying Lemke's algorithm to the linear complementarity problem induced by the sequence-form representation; the practical efficiency of this approach was demonstrated in the Gala system~\cite{KollerRepresentationssolutionsgametheoretic1997}. For $n$-player games, Govindan and Wilson~\cite{govindanStructureTheoremsGame2002} extended structure theorems to perturbed extensive-form games through enabling strategies, which are closely related to sequence-form strategies, and obtained a piecewise-differentiable path-following method for computing Nash equilibria. For two-player zero-sum extensive-form games,
	Ling et al.~\cite{lingWhatGameAre2018}
	established that the dilated-entropy-regularized
	sequence-form formulation yields solutions
	that are realization-equivalent to
	the logistic QREs of the reduced normal form,
	and developed a primal--dual Newton method
	for computing these equilibria. More recently, Hou et al.~\cite{houSequenceformDifferentiablePathfollowing2025} developed a globally differentiable sequence-form path-following method for Nash equilibrium computation by incorporating logarithmic-barrier terms into the payoff functions. Computational methods based on the sequence form have also been extended to normal-form equilibrium refinements. Extending the method of van den Elzen and Talman~\cite{vandenelzenProcedureFindingNash1991} to the extensive-form game, von Stengel et al.~\cite{vonStengelComputingNormalForm2002} established a piecewise linear sequence-form path for computing normal-form perfect equilibria in two-player games. More recently, Hou et al.~\cite{houSequenceFormCharacterizationDifferentiable2025,houCharacterizationComputationNormalForm2026a} derived sequence-form characterizations of normal-form perfect and proper equilibria and, on this basis, developed differentiable path-following methods applicable to $n$-player games.
	
	Nevertheless, existing sequence-form methods
	for logistic QRE and normal-form equilibrium
	refinements do not directly provide
	a sequence-form characterization of
	the perturbed logistic-QRE variant.
	The main difficulty lies in expressing
	the entropy-regularized logit responses,
	originally defined over pure-strategy
	distributions in the reduced normal form,
	in terms of realization plans while
	preserving the full-support reference
	profile and the corresponding equilibrium
	conditions.
	To address this difficulty,
	we construct an entropy-regularized
	artificial game in sequence form
	and establish its equivalence
	with the normal-form logistic-QRE variant.
	Based on this formulation,
	we develop a differentiable
	path-following method starting from
	a prescribed
	interior realization plan.
	We establish the existence of a smooth
	solution path for almost all sufficiently
	small perturbations and prove that
	its limiting mixed-strategy profiles
	are normal-form perfect.
	Furthermore, by reformulating
	the dilated-entropy terms as weighted
	standard entropy terms,
	we obtain an equivalent smooth system
	that generates identical
	realization-plan paths.
	These developments yield two
	sequence-form path-following methods. When the strategy perturbation vanishes
	and the reference profile is uniform,
	these methods can also be applied
	to efficiently computing standard logistic QREs
	in $n$-player extensive-form games.
	Numerical experiments illustrate
	their equilibrium-selection behavior
	and compare their computational performance
	with that of the original normal-form method.
	
	The remainder of this paper is organized
	as follows.
	Section~\ref{sec:preliminaries-and-motivation}
	introduces the game representations
	and reviews the logistic-QRE variant
	for selecting normal-form perfect equilibria.
	Section~\ref{sec:sequence-form-formulation}
	establishes an equivalent sequence-form
	characterization of the variant.
	Section~\ref{sec:equilibrium-path}
	develops a differentiable path-following
	method and establishes the existence
	and limiting properties of the solution path.
	Section~\ref{sec:equivalent-smooth-path}
	derives an equivalent smooth system
	by reformulating the dilated-entropy terms
	as weighted standard entropy terms.
	Section~\ref{sec:numerical-experiments}
	presents numerical experiments
	comparing the proposed sequence-form methods
	with the original normal-form method.
	Section~\ref{sec:conclusion}
	concludes the paper.
	
\clearpage
\section{Preliminaries and Computational Motivation}\label{sec:preliminaries-and-motivation}
\subsection{Game Representations and the Reduced Normal Form}
	\begin{table}[tb!]
		\centering\small\captionsetup{skip=-8pt}
			\caption{Notation for Extensive-Form Games and the Reduced Normal Form}
		\begin{tabularx}{\textwidth}{@{}p{0.28\textwidth}>{\raggedright\arraybackslash}X@{}}
			\toprule
			Symbol & Explanation\\
			\midrule
			$N=\{1,2,\ldots,n\}$ & Set of players\\
			$N_c=N\cup\{c\}$ & Set of players and chance player $c$\\
			$a$ & Action taken by a player\\
			$H$ & Set of histories,  $\emptyset\in H$ and $\langle a_1,\ldots,a_L\rangle\in H$ if $\langle a_1,\ldots,a_K\rangle\in H$ and $L<K$\\
			$Z$ & Set of terminal histories\\
			$A(h)$ & Set of actions after a nonterminal history $h$\\
			$P(h)$ & Player who takes an action after $h$\\
			$f_{c}(a|h)$ & Probability that chance player $c$ takes action $a$ after $h$\\
			$-i$ & All non-chance players excluding player $i\in N$\\
			$\mathcal{I}_{i}$ & Collection of information partitions of $\{h\in H\mid P(h)=i\}$\\
			$M_{i}=\{1,\ldots,m_{i}\}$ & Set of information partition indices for player $i\in N_c$\\
			$I^{j}_{i}\in\mathcal{I}_{i},j\in M_{i}$ & $j$th information set of player $i\in N_c$, $A(I^j_i)\triangleq A(h)= A(h')$ whenever $h,h'\in I^j_i$\\
			$\succsim_i$ & Preference relation of player $i\in N$ \\
			$u_z^{i}:Z\to\mathbb{R}$ & Payoff function of player $i\in N$\\
			$R_{i}(h)$ & Record of player $i\in N_c$'s experience along $h$\\
			$|C|$ & Cardinality of a finite set $C$\\
			$s^i$ & Pure strategy of player $i$\\
			$S=\underset{i\in N_c}{\prod}S^i$ & Set of pure-strategy profiles\\
			$u^i(s)$ & Expected payoff of player $i$ on the pure-strategy profile $s\in S$\\
			$\bm{\sigma}^i$ & Mixed strategy of player $i\in N_c$, probability measure over $S^i$\\
			$\Xi=\underset{i\in N}{\prod}\Xi^i$ & Set of mixed-strategy profiles,  $\Xi^i=\{\bm{\sigma}^i:S^i\to\mathbb{R}_+\mid \sum\limits_{s^i\in S^i}\sigma^{i}(s^i)=1\}$\\
			$\Xi_{++}=\underset{i\in N}{\prod} \Xi^i_{++}$ & Set of strictly positive mixed-strategy profiles\\
			\bottomrule
		\end{tabularx}
		\label{tab:normal-form-notation}
	\end{table}
	Following Osborne and Rubinstein~\cite{OsborneCourseGameTheory1994}, we represent an extensive-form game by $\Gamma=\langle N, H, P, f_c, \{{\cal I}_i\}_{i\in N}, \{\succsim_i\}_{i\in N}\rangle$, with notation summarized in Table~\ref{tab:normal-form-notation}. Throughout this paper, we restrict attention to finite extensive-form games with perfect recall. Finiteness requires $H$ to be finite, whereas perfect recall requires each player to remember all previously observed information sets and actions taken. Formally, for every $i\in N$ and $I_i^j\in\mathcal{I}_i$,
	we have $R_i(h)=R_i(h')$ for all $h,h'\in I_i^j$.
	
	The normal-form representation of $\Gamma$ is given by $\Gamma_n=\langle N, S, \bm{\sigma}^c, \{u^i\}_{i\in N}\rangle$, with the associated notation summarized in Table~\ref{tab:normal-form-notation}. $\bm{\sigma}^c$ denotes the fixed chance strategy induced by $f_c$. For each player $i\in N_c$, a pure strategy is a function $s^i$ that assigns an action in $A(I_i^j)$ to every information set $I_i^j$, $j\in M_i$. Consequently, the number of such pure strategies is $\prod_{j\in M_i}|A(I_i^j)|$, which can grow exponentially with the number of information sets. We instead consider the more compact reduced normal form, in which a pure strategy $s^i$ prescribes an action at $I_i^j$ only when that information set is reachable under its preceding prescriptions. Nevertheless, the number of pure strategies in the reduced normal form may still grow exponentially with the number of parallel information sets, posing a substantial computational challenge. To facilitate computation, given a pure strategy $s^i$ of player $i\in N_c$, let $s^i(a)$ equal $1$ if $s^i$ prescribes action $a$, and $0$ otherwise. Then, for any pure-strategy profile \(s=(s^i:i\in N_c)\), the payoff of player \(i\in N\) is given by $u^i(s)=\sum_{h=\langle a_1,\ldots,a_L\rangle\in Z}u^i_z(h)\prod_{q=0}^{L-1}s^{P(\langle a_1,\ldots,a_q\rangle)}(a_{q+1})$. Given a mixed-strategy profile $\bm{\sigma}=(\bm{\sigma}^i:i\in N)\in \Xi$, the expected payoff of player $i\in N$ is $u^i(\bm{\sigma})=\sum_{s^i\in S^i}\sigma^i(s^i)u^i(s^i,\bm{\sigma}^{-i})$ with $u^i(s^i,\bm{\sigma}^{-i})=\sum_{s^{-i}\in S^{-i}}u^i(s^i,s^{-i})\prod_{i_q\in N_c\backslash \{i\}}\sigma^{i_q}(s^{i_q})$. A mixed-strategy profile $\bm{\sigma}^*$ is a Nash equilibrium if, for every player $i\in N$ and $s^i\in S^i$, it holds that $\sigma^{*i}(s^i)=0$ whenever $u^i(s^i,\bm{\sigma}^{*-i})< u^i(\tilde s^i,\bm{\sigma}^{*-i})$ for some $\tilde s^i\in S^i$. Selten~\cite{SeltenReexaminationperfectnessconcept1975} introduced the concept of perfect equilibrium, eliminating a large number of unreasonable Nash equilibria. The definition of normal-form perfect equilibrium in an extensive-form game is as follows.
	\begin{definition}\label{def:normal-form-perfect-equilibrium}
		Let $\Gamma$ be an extensive form game. For any sufficiently small $\varepsilon > 0$, a totally mixed strategy profile $\bm{\sigma}(\varepsilon)\in \Xi$ is an $\varepsilon$-normal-form perfect equilibrium of $\Gamma$ if $\sigma^i(\varepsilon;s^i)\leq\varepsilon$ whenever $u^i(s^i, \bm{\sigma}^{-i}(\varepsilon)) < u^i(\tilde s^{i},\bm{\sigma}^{-i}(\varepsilon))$ for all $i\in N$ and $s^i,\tilde s^i\in S^i$. A mixed strategy profile $\bm{\sigma}^*\in \Xi$ is defined as a normal-form perfect equilibrium of game $\Gamma$ if $\bm{\sigma}^*$ is a limit point of some sequence $\{\bm{\sigma}(\varepsilon^k)\}_{k=1}^\infty$, where $\lim_{k\to\infty}\varepsilon^k=0$ and each $\bm{\sigma}(\varepsilon^k)$ is an $\varepsilon^k$-normal form perfect equilibrium of $\Gamma$.
	\end{definition}
    Although normal-form perfection refines Nash equilibrium, a game may still admit multiple normal-form perfect equilibria, necessitating an additional criterion for selecting among them. Logistic QRE~\cite{mckelveyQuantalResponseEquilibria1995} provides a systematic selection mechanism based on payoff-sensitive probabilistic responses.

\subsection{A Logistic QRE Variant for Selecting Normal-Form Perfect Equilibria}

The limiting equilibria of standard logistic QRE need not be normal-form perfect. To combine the equilibrium-selection property of logistic QRE with the refinement property of normal-form perfection, Cao et al.~\cite{caoVariantLogisticQuantal2024} developed a perturbed logistic QRE variant whose limiting equilibria satisfy the requirements of normal-form perfection. We review this construction and its equivalent entropy-regularized formulation as the basis for the subsequent sequence-form reformulation.

Fix $\varepsilon_0\in(0,1)$ and let $\bm{\eta^0_s}=(\eta_s^{0i}(s^i):i\in N,s^i\in S^i)$ be a perturbation vector satisfying $\eta_s^{0i}(s^i)>0$ and $\sum_{s_q^i\in S^i}\eta_s^{0i}(s_q^i)=\varepsilon_0<1$ for every $i\in N$. For any $\bm{\sigma}\in\Xi$ and $t\in [0,1]$, define the mixed-strategy perturbation mapping $\bm{\pi_s}(\bm{\sigma},t)=(\pi^i_s(\bm{\sigma^i},t;s^i):i\in N,s^i\in S^i)$ by $\pi^i_s(\bm{\sigma^i},t;s^i)=(1-t\varepsilon_0)\sigma^i(s^i)+t\eta_s^{0i}(s^i)$. The mapping preserves nonnegativity and normalization, and hence $\bm{\pi_s}(\bm{\sigma},t)\in\Xi$. Moreover, $\bm{\pi_s}(\bm{\sigma},t)\in\Xi_{++}$ for every $t>0$, ensuring that all pure strategies receive strictly positive probability.

Let $\bm{\sigma^0}=(\sigma^{0i}(s^i):i\in N,s^i\in S^i)$ be a prescribed completely mixed strategy profile satisfying $\sigma^{0i}(s^i)>\eta_s^{0i}(s^i)$ for all $i\in N$ and $s^i\in S^i$. Set $c(t)=\exp(1-1/t)$ for $t\in(0,1]$, with $c(0)=0$, and define $\lambda_p(t)=(1-c(t))/c(t)$. Then $\lambda_p(t)$ is strictly decreasing on $(0,1]$, with $\lambda_p(1)=0$ and $\lim_{t\to0^+}\lambda_p(t)=+\infty$. For $t\in(0,1]$, a mixed-strategy profile
$\bm{\sigma}(\lambda_p(t))\in\Xi$
is called a logistic QRE variant at rationality parameter
$\lambda_p(t)$ if, for every $i\in N$ and $s^i\in S^i$,
\begin{equation}\label{eq:logit-qre-variant}
	\sigma^i(\lambda_p(t);s^i)
	=
	\frac{
		\sigma^{0i}(s^i)
		\exp\left(
		\lambda_p(t)
		u^i\left(
		s^i,
		\bm{\pi_s}(
		\bm{\sigma^{-i}}(\lambda_p(t)),t)
		\right)
		\right)
	}{
		\sum\limits_{s_q^i\in S^i}
		\sigma^{0i}(s_q^i)
		\exp\left(
		\lambda_p(t)
		u^i\left(
		s_q^i,
		\bm{\pi_s}(
		\bm{\sigma^{-i}}(\lambda_p(t)),t)
		\right)
		\right)
	}.
\end{equation}
At $t=1$, $\lambda_p(t)=0$, and the response function above gives $\bm{\sigma}(0)=\bm{\sigma^0}$. As $t$ decreases from $1$ to $0$, $\lambda_p(t)$ increases and the players' strategy probabilities increasingly concentrate on pure strategies with higher expected payoffs. Unlike the standard logistic QRE, this variant
evaluates pure-strategy payoffs against the perturbed
mixed-strategy profile
$\bm{\pi_s}(\bm{\sigma^{-i}}(\lambda_p(t)),t)$
of the other players.
For a uniform reference profile,
$\sigma^{0i}(s^i)=1/|S^i|$, and vanishing perturbations,
$\varepsilon_0\to0^+$,
Equation~\eqref{eq:logit-qre-variant} reduces to
the standard logistic QRE for each fixed $t\in(0,1]$.
The variant thus incorporates explicit strategy perturbations
into the quantal-response framework,
providing a basis for selecting normal-form perfect equilibria.

To clearly describe this selection procedure, consider an entropy-barrier artificial game $\Gamma_p^e(t)$ for $t\in(0,1]$ in mixed-strategy space. Given $\bm{\hat{\sigma}^{-i}}\in\Xi^{-i}$, player $i$ obtains a best-response strategy by solving the following concave optimization problem:
\begin{equation}\label{eq:normal-form-best-response}
	\begin{aligned}
		\max\limits_{\bm{\sigma^i}}\quad
		&
		(1\!-\! c(t))\sum\limits_{s^i\in S^i}\sigma^i(s^i)u^i(s^i,\bm{\pi_s}(\bm{\hat{\sigma}^{-i}},t))\!-\! c(t)\sum\limits_{s^i\in S^i}\sigma^i(s^i)\big(\ln\sigma^i(s^i)\!-\!\ln\sigma^{0i}(s^i)\big)\\
		\text{s.t.}\quad&\sum\limits_{s^i\in S^i}\sigma^i(s^i)=1.
	\end{aligned}
\end{equation}
The first term in the objective represents the expected payoff under the perturbed strategy profile, while the second is a relative-entropy regularizer with respect to the reference strategy $\bm{\sigma^{0i}}$. Applying the first-order optimality conditions to each player's best-response problem~\eqref{eq:normal-form-best-response} and imposing the fixed-point condition $\bm{\hat{\sigma}}=\bm{\sigma}$ yields the following equilibrium system for $\Gamma_p^e(t)$:
\begin{equation}\label{eq:normal-form-equilibrium-system}
	\begin{array}{l}
		(1-c(t))u^i(s^i,\bm{\pi_s}(\bm{\hat{\sigma}^{-i}},t))-c(t)\big(\ln\sigma^i(s^i)-\ln\sigma^{0i}(s^i)\big)\\
		\hspace{6cm}-c(t)-\nu^i=0,\; i\in N, s^i\in S^i,\\
		\sum\limits_{s^i\in S^i}\sigma^i(s^i)-1=0,\; i\in N,\quad 0<\sigma^i(s^i),\; i\in N,s^i\in S^i.
	\end{array}
\end{equation}
For each $t\in(0,1]$, a mixed-strategy profile
$\bm{\sigma^*}\in\Xi$ is a Nash equilibrium of
$\Gamma_p^e(t)$ if and only if there exists
$\bm{\nu^*}=(\nu^{*i}:i\in N)$ such that
$(\bm{\sigma^*},\bm{\nu^*},t)$ satisfies
System~\eqref{eq:normal-form-equilibrium-system}.
Eliminating the multipliers yields
Equation~\eqref{eq:logit-qre-variant} with rationality
parameter $\lambda_p(t)$.
Thus, the Nash equilibria of $\Gamma_p^e(t)$
coincide with the logistic QRE variants of $\Gamma$,
and System~\eqref{eq:normal-form-equilibrium-system} characterizes
their parameterized equilibrium correspondence.
At $t=1$, the unique equilibrium is $\bm{\sigma^0}$;
as $t\to0^+$, $\lambda_p(t)\to+\infty$.
The following lemma establishes that every accumulation
point of the variant equilibria as $t\to0^+$
is a normal-form perfect equilibrium of $\Gamma$.
\begin{lemma}\label{lem:variant-perfect-limit}
	Let $\{t_k\}_{k=1}^{\infty}\subset(0,1]$
	satisfy $\lim_{k\to\infty}t_k=0$,
	and let $\bm{\sigma^*}(t_k)$
	be a Nash equilibrium of
	$\Gamma_p^e(t_k)$ for each $k$.
	Then every accumulation point of
	$\{\bm{\sigma^*}(t_k)\}_{k=1}^{\infty}$
	is a normal-form perfect equilibrium of
	the original game $\Gamma$.
\end{lemma}

However, both the number of variables and the number of constraints in System~\eqref{eq:normal-form-equilibrium-system} grow exponentially with the size of the extensive-form game. Directly computing the normal-form perfect equilibrium selected by the above Logistic QRE variant path in mixed-strategy space is therefore generally expensive. As illustrated in Fig.~\ref{fig:variable-count-comparison}, the number of variables in the normal-form system grows rapidly as the game parameters increase, whereas the sequence-form systems exhibit substantially slower growth. This motivates the development of an equivalent sequence-form formulation of the logistic QRE variant in the next section. This motivates the development of an equivalent and compact sequence-form formulation of the logistic QRE variant, as presented in the next section.

\section{A Sequence-Form Formulation of the Logistic QRE Variant}\label{sec:sequence-form-formulation}

\begin{table}[tb!]
	\caption{Additional notation for the sequence form}\label{tab:sequence-form-notation}
	\centering\small
	\begin{tabular}{ll}
		\toprule
		Symbol & Explanation\\
		\midrule
		$\varpi^i$ & Sequence of actions taken by player $i$\\
		$\varpi^i_{I^j_i}$ & Sequence of player $i$ leading to $I^j_i$\\
		$\varpi^i_{I^j_i}a$ & Sequence $\varpi^i_{I^j_i}$ extended by action $a$\\
		${W}=\underset{i\in N_c}\prod{W}^i$ & Set of sequence profiles, $\emptyset\in{W}^i$\\
		$g^i(\varpi)$ & Payoff associated with sequence profile $\varpi$\\
		$\bm{\gamma}^i$ & Realization plan of player $i\in N_c$\\
		$\Lambda=\underset{i\in N}\prod\Lambda^i$ & Set of realization-plan profiles\\
		$\Lambda_{++}=\underset{i\in N}\prod\Lambda^i_{++}$ & Set of strictly positive realization-plan profiles\\
		$M_i(\varpi^i)$ & Information-set indices with parent sequence $\varpi^i$\\
		$m_i(\varpi^i)$ & Cardinality $|M_i(\varpi^i)|$\\
		$m_0$ & Number of non-chance information sets, $\sum_{i\in N}m_i$\\
		$n_0$ & Number of non-chance actions, $\sum_{i\in N}\sum_{j\in M_i}|A(I^j_i)|$\\
		\bottomrule
	\end{tabular}
\end{table}
The sequence-form representation of $\Gamma$ is denoted by $\Gamma_s=\langle N,W,\bm{\gamma}^c,\{g^i\}_{i\in N}\rangle$ with the relevant notation summarized in Table~\ref{tab:sequence-form-notation}. We say that $\varpi=(\varpi^i:i\in N_c)\in W$ is induced by a history $h\in H$ if $\varpi^i=\varpi_h^i$ for every $i\in N_c$. For each player $i\in N$, the payoff function $g^i$ is defined by $g^i(\varpi)=u_z^i(h)$ if $\varpi$ is induced by a terminal history $h\in Z$, and $g^i(\varpi)=0$ otherwise. For each player $i\in N_c$, a realization plan in the sequence form is a function $\bm{\gamma}^i$ defined on ${W}^i$ satisfying $\gamma^i(\emptyset)=1$ and the following flow constraints,
\begin{equation}\label{eq:realization-flow-constraints}
	\begin{aligned}
		& \sum\limits_{a\in A(I_i^j)}\gamma^i(\varpi_{I_i^j}^i a)-\gamma^i(\varpi_{I_i^j}^i)=0,\; j\in M_i,\\
		& 0\le \gamma^i(\varpi_{I_i^j}^i a),\; j\in M_i,a\in A(I_i^j).
	\end{aligned}
\end{equation}
Given $\bm{\gamma}=(\bm{\gamma}^i:i\in N)\in\Lambda$, $g^i(\bm{\gamma})
=
\sum_{\varpi^i\in W^i}
\gamma^i(\varpi^i)
g^i(\varpi^i,\bm{\gamma}^{-i})$, where $g^i(\varpi^i,\bm{\gamma}^{-i})=\sum_{\varpi^{-i}\in {W}^{-i}}g^i(\varpi^i,\varpi^{-i})\prod_{i_q\ne i}\gamma^{i_q}(\varpi^{i_q})$, gives player $i$'s expected payoff. The number of sequences available to player $i\in N$ is $\sum_{j \in M_i} |A(I_i^j)|+1$, and hence grows linearly with the number of information sets.

Given any pure strategy $s^i\in S^i$ of player $i\in N_c$, define $s^i(\varpi^i)=\prod_{a\in \varpi^i}s^i(a)$ for $\varpi^i\in W^i$. For any $\bm{\sigma}\in \Xi$, let $\bm{\gamma}(\bm{\sigma})=(\gamma^{i}(\bm{\sigma}^i;\varpi^i):i\in N_c,\varpi^i\in W^i)$, where $\gamma^{i}(\bm{\sigma}^i;\varpi^i)=\sum_{s^i\in S^i}s^i(\varpi^i)\sigma^i(s^i),\,i\in N_c,\varpi^i\in W^i$. It follows that $\gamma^{i}(s^i;\varpi^i)=s^i(\varpi^i)$ and $\bm{\gamma}(\bm{\sigma})\in\Lambda$. Define $T=\left\{(\bm{\sigma},\bm{\gamma})\left|\bm{\sigma}\in  \Xi,\bm{\gamma}=\bm{\gamma}(\bm{\sigma})\right.\right\}$, which leads to the following conclusions. For any $\bm{\gamma}\in\Lambda$, there exists $\bm{\sigma}\in\Xi$ such that $(\bm{\sigma},\bm{\gamma})\in T$. Specially, for any \(\bm{\gamma} \in \Lambda_{++}\), one such mixed-strategy profile is given by $\bm{\sigma}(\bm{\gamma})=(\bm{\sigma}^i(\bm{\gamma}^i): i\in N_c)$ with \begin{equation}\label{eq:canonical-mixed-strategy}\sigma^i(\bm{\gamma}^i;s^i)	= \prod_{j\in M_i,\, a\in A(I^j_i):\,s^i(a)=1}\frac{\gamma^i(\varpi^i_{I^j_i}a)}{\gamma^i(\varpi^i_{I^j_i})}.\end{equation}
Moreover, if $(\bm{\sigma},\bm{\gamma})\in T$, then $u^i(\bm{\sigma})=g^i(\bm{\gamma})$ for every player $i\in N$. Detailed proofs of these results can be found in Hou et al.~\cite{houSequenceFormCharacterizationDifferentiable2025}.

We next express the mixed-strategy perturbation in the sequence form. Define $\bm{\eta^0_w}=(\eta_w^{0i}(\varpi^i):i\in N,\varpi^i\in W^i)$ by
\begin{equation}\label{eq:normal-to-sequence-tremble}
	\eta^{0i}_w(\varpi^i)=\sum\limits_{s^i_q\in S^i:s^i_q(\varpi^i)=1}\eta_s^{0i}(s^i_q),\; i\in N,\varpi^i\in W^i.
\end{equation}
By construction, $\bm{\eta^0_w}$ satisfies
\begin{equation}\label{eq:sequence-tremble-flow}
	\begin{aligned}
		&\sum\limits_{a\in A(I^j_i)}\eta^{0i}_w(\varpi^i_{I^j_i}a)-\eta^{0i}_w(\varpi^i_{I^j_i})=0,\; i\in N,j\in M_i,\\
		&0<\eta^{0i}(\varpi^i_{I^j_i}a)\leq \varepsilon_0,\; i\in N,j\in M_i,a\in A(I^j_i),
	\end{aligned}
\end{equation}
with $\eta^{0i}_w(\varpi^i_\emptyset)=\varepsilon_0$ for every $i\in N$. For any $\bm{\gamma}\in\Lambda$ and $t\in[0,1]$, define the realization-plan perturbation mapping $\bm{\pi_w}(\bm{\gamma},t)=(\pi^i_w(\bm{\gamma^i},t;\varpi^i):i\in N,\varpi^i\in W^i)$ by $\pi^i_w(\bm{\gamma^i},t;\varpi^i)=(1-t\varepsilon_0)\gamma^i(\varpi^i)+t\eta_w^{0i}(\varpi^i)$. The flow constraints in Equation~\eqref{eq:sequence-tremble-flow} ensure that $\bm{\pi_w}(\bm{\gamma},t)\in\Lambda$. Furthermore, the mixed-strategy
and realization-plan perturbations are consistent: for every $\bm{\sigma}\in\Xi$,
\begin{equation}\label{eq:perturbation-map-consistency}
	\bm{\gamma}(\bm{\pi_s}(\bm{\sigma},t))=\bm{\pi_w}(\bm{\gamma}(\bm{\sigma}),t).
\end{equation}
Although the mapping according to $T$ from mixed strategies to realization plans is generally nonunique, the logistic QRE variant admits a unique
recovery when the reference profile satisfies $\bm{\sigma^0}=\bm{\sigma}(\bm{\gamma^0})$ for some $\bm{\gamma}^0\in\Lambda_{++}$. The following lemma establishes this property.

\begin{lemma}\label{lem:canonical-variant-strategy}
	Let $(\bm{\sigma}^*,\bm{\gamma}^*)\in T$ and suppose that
	$\bm{\sigma^0}=\bm{\sigma}(\bm{\gamma^0})$,
	where $\bm{\gamma^0}\in\Lambda_{++}$.
	If $\bm{\sigma}^*$ satisfies the logistic
	QRE variant~\eqref{eq:logit-qre-variant}
	for some $t\in(0,1]$,
	then
	$\bm{\sigma}^*
	=\bm{\sigma}(\bm{\gamma}^*)$.
\end{lemma}

\begin{proof}
	Fix an arbitrary player $i\in N$.
	Since $\bm{\sigma}^{0i}$ has full support,
	Equation~\eqref{eq:logit-qre-variant} implies that
	$\sigma^{*i}(s^i)>0$
	for every $s^i\in S^i$.
	Consequently,
	$\bm{\gamma}^{*i}\in\Lambda_{++}^i$. Define recursively
	\begin{equation}\label{eq:partition-function-recursion}
		\begin{aligned}
			\mathcal{Z}^i(\bm{\gamma}^*;\varpi^i)
			={}&
			\exp\left(
			\lambda_p(t)
			g^i\left(
			\varpi^i,
			\bm{\pi_w}(\bm{\gamma}^{*-i},t)
			\right)
			\right)
			\prod_{j\in M_i(\varpi^i)}
			\mathcal{Z}^i(\bm{\gamma}^*;I_i^j),\;\varpi^i\in W^i,
			\\
			\mathcal{Z}^i(\bm{\gamma}^*;I_i^j)
			={}&
			\sum_{a\in A(I_i^j)}
			\frac{
				\gamma^{0i}(\varpi^i_{I_i^j}a)
			}{
				\gamma^{0i}(\varpi^i_{I_i^j})
			}
			\mathcal{Z}^i(
			\bm{\gamma}^*;
			\varpi^i_{I_i^j}a
			),
			\quad j\in M_i.
		\end{aligned}
	\end{equation}
	Using $\bm{\sigma}^{0i}=\bm{\sigma}^i(\bm{\gamma}^{0i})$
	and the recursive definition~\eqref{eq:partition-function-recursion},
	we obtain
	\begin{equation}\label{eq:partition-function-normalizer}
		\begin{aligned}
			\mathcal{Z}^i(
			\bm{\gamma}^*;
			\varpi^i_{\emptyset}
			)
			&=
			\sum_{s^i\in S^i}
			\sigma^{0i}(s^i)
			\exp\left(
			\lambda_p(t)
			\sum_{
				\varpi_q^i\in W^i:
				s^i(\varpi_q^i)=1
			}
			g^i\bigl(
			\varpi_q^i,
			\bm{\pi_w}(\bm{\gamma}^{*-i},t)
			\bigr)
			\right)
			\\
			&=
			\sum_{s^i\in S^i}
			\sigma^{0i}(s^i)
			\exp\left(
			\lambda_p(t)
			u^i\bigl(
			s^i,
			\bm{\pi_s}(\bm{\sigma}^{*-i},t)
			\bigr)
			\right).
		\end{aligned}
	\end{equation}	
	Since $(\bm{\sigma}^*,\bm{\gamma}^*)\in T$, the induced realization plan satisfies
	\begin{equation}\label{eq:realization-from-mixed-strategy}
		\gamma^{*i}(\varpi^i)
		=
		\sum_{
			s^i\in S^i:
			s^i(\varpi^i)=1
		}
		\sigma^{*i}(s^i),
		\; \varpi^i\in W^i.
	\end{equation}	
	Moreover, the logistic QRE variant gives
	\begin{equation}\label{eq:normal-form-logit-response}
		\sigma^{*i}(s^i)
		=
		\frac{
			\sigma^{0i}(s^i)
			\exp\left(
			\lambda_p(t)
			u^i\bigl(
			s^i,
			\bm{\pi_s}(\bm{\sigma}^{*-i},t)
			\bigr)
			\right)
		}{
			\sum\limits_{s_q^i\in S^i}
			\sigma^{0i}(s_q^i)
			\exp\left(
			\lambda_p(t)
			u^i\bigl(
			s_q^i,
			\bm{\pi_s}(\bm{\sigma}^{*-i},t)
			\bigr)
			\right)
		}.
	\end{equation}
	Combining
	Equations~\eqref{eq:realization-from-mixed-strategy}
	and~\eqref{eq:normal-form-logit-response},
	for every $j\in M_i$ and
	$a\in A(I_i^j)$, we obtain
	\begin{equation}\label{eq:canonical-action-ratio}
		\begin{aligned}
			\frac{
				\gamma^{*i}(\varpi^i_{I_i^j}a)
			}{
				\gamma^{*i}(\varpi^i_{I_i^j})
			}
			&=
			\frac{
				\sum\limits_{
					s_q^i\in S^i:
					s_q^i(\varpi^i_{I_i^j}a)=1
				}
				\sigma^{0i}(s_q^i)
				\exp\left(
				\lambda_p(t)
				u^i\bigl(
				s_q^i,
				\bm{\pi_s}(\bm{\sigma}^{*-i},t)
				\bigr)
				\right)
			}{
				\sum\limits_{
					s_q^i\in S^i:
					s_q^i(\varpi^i_{I_i^j})=1
				}
				\sigma^{0i}(s_q^i)
				\exp\left(
				\lambda_p(t)
				u^i\bigl(
				s_q^i,
				\bm{\pi_s}(\bm{\sigma}^{*-i},t)
				\bigr)
				\right)
			}
			\\
			&=
			\frac{
				\gamma^{0i}(\varpi^i_{I_i^j}a)
			}{
				\gamma^{0i}(\varpi^i_{I_i^j})
			}
			\frac{
				\mathcal{Z}^i(
				\bm{\gamma}^*;
				\varpi^i_{I_i^j}a
				)
			}{
				\mathcal{Z}^i(
				\bm{\gamma}^*;
				I_i^j
				)
			},
		\end{aligned}
	\end{equation}
	where the equality follows from
	$\bm{\sigma}^{0i}=\bm{\sigma}^i(\bm{\gamma}^{0i})$
	and the recursive definition~\eqref{eq:partition-function-recursion}.
	
	Substituting~\eqref{eq:canonical-action-ratio}
	into~\eqref{eq:canonical-mixed-strategy}
	and applying~\eqref{eq:partition-function-normalizer}
	yields, for every $s^i\in S^i$,
	\begin{equation}\label{eq:canonical-logit-identity}
		\begin{aligned}
			\sigma^i(\bm{\gamma}^{*i};s^i)
			&=
			\prod_{
				j\in M_i,\,
				a\in A(I_i^j):
				s^i(\varpi^i_{I_i^j}a)=1
			}
			\frac{
				\gamma^{*i}(\varpi^i_{I_i^j}a)
			}{
				\gamma^{*i}(\varpi^i_{I_i^j})
			}
			\\
			&=
			\prod_{
				j\in M_i,\,
				a\in A(I_i^j):
				s^i(\varpi^i_{I_i^j}a)=1
			}
			\left(
			\frac{
				\gamma^{0i}(\varpi^i_{I_i^j}a)
			}{
				\gamma^{0i}(\varpi^i_{I_i^j})
			}
			\frac{
				\mathcal{Z}^i(
				\bm{\gamma}^*;
				\varpi^i_{I_i^j}a
				)
			}{
				\mathcal{Z}^i(
				\bm{\gamma}^*;
				I_i^j
				)
			}
			\right)
			\\
			&=
			\frac{
				\sigma^{0i}(s^i)
				\exp\left(
				\lambda_p(t)
				\sum\limits_{
					\varpi_q^i\in W^i:
					s^i(\varpi_q^i)=1
				}
				g^i\bigl(
				\varpi_q^i,
				\bm{\pi_w}(\bm{\gamma}^{*-i},t)
				\bigr)
				\right)
			}{
				\mathcal{Z}^i(
				\bm{\gamma}^*;
				\varpi^i_{\emptyset}
				)
			}
			\\
			&=
			\frac{
				\sigma^{0i}(s^i)
				\exp\left(
				\lambda_p(t)
				u^i\bigl(
				s^i,
				\bm{\pi_s}(\bm{\sigma}^{*-i},t)
				\bigr)
				\right)
			}{
				\sum\limits_{s_q^i\in S^i}
				\sigma^{0i}(s_q^i)
				\exp\left(
				\lambda_p(t)
				u^i\bigl(
				s_q^i,
				\bm{\pi_s}(\bm{\sigma}^{*-i},t)
				\bigr)
				\right)
			}=
			\sigma^{*i}(s^i).
		\end{aligned}
	\end{equation}	
	Since $i\in N$ was arbitrary,
	we conclude that
	$\bm{\sigma}^*=\bm{\sigma}(\bm{\gamma}^*)$.
	This completes the proof.
\end{proof}
When the reference profile satisfies
$\bm{\sigma^0}=\bm{\sigma}(\bm{\gamma^0})$,
the correspondence in~\eqref{eq:canonical-mixed-strategy}
allows the mixed-strategy relative-entropy term
to be expressed equivalently in terms of realization plans.
For $\bm{\gamma}^i\in\Lambda_{++}^i$, we have 
\[
\begin{aligned}
	&\sum_{s^i\in S^i}
	\sigma^i(\bm{\gamma}^i;s^i)
	\ln
	\frac{\sigma^i(\bm{\gamma}^i;s^i)}
	{\sigma^{0i}(s^i)}
	\\
	&\hspace{0.3cm}=
	\sum_{j\in M_i}
	\sum_{a\in A(I_i^j)}
	\left(
	\sum_{{s^i\in S^i:
			s^i(\varpi^i_{I_i^j}a)=1}}
	\sigma^i(\bm{\gamma}^i;s^i)
	\ln
	\frac{
		\gamma^i(\varpi^i_{I_i^j}a)
		\gamma^{0i}(\varpi^i_{I_i^j})
	}{
		\gamma^i(\varpi^i_{I_i^j})
		\gamma^{0i}(\varpi^i_{I_i^j}a)
	}
	\right)
	\\
	&\hspace{0.3cm}=
	\sum_{j\in M_i}
	\sum_{a\in A(I_i^j)}
	\gamma^i(\varpi^i_{I_i^j}a)
	\Bigl(
	\ln\gamma^i(\varpi^i_{I_i^j}a)
	-
	\ln\gamma^i(\varpi^i_{I_i^j})
	-
	\ln\gamma^{0i}(\varpi^i_{I_i^j}a)
	+
	\ln\gamma^{0i}(\varpi^i_{I_i^j})
	\Bigr).
\end{aligned}
\]
This identity leads to a corresponding
dilated-entropy-regularized game
$\Gamma_s^e(t)$ in the sequence form.
Given a realization-plan profile
$\bm{\hat{\gamma}}\in\Lambda$,
each player $i\in N$ determines
an optimal response by solving
\begin{equation}\label{eq:sequence-form-best-response}
	\begin{aligned}
		\max_{\bm{\gamma^i}}
		&\quad
		(1-c(t))
		\sum_{j\in M_i}
		\sum_{a\in A(I_i^j)}
		\gamma^i(\varpi^i_{I_i^j}a)
		g^i\bigl(
		\varpi^i_{I_i^j}a,
		\bm{\pi_w}(\bm{\hat\gamma^{-i}},t)
		\bigr)
		\\
		&\quad-
		c(t)\sum_{j\in M_i}
		\sum_{a\in A(I_i^j)}
		\gamma^i(\varpi^i_{I_i^j}a)
		\Bigl(
		\ln\gamma^i(\varpi^i_{I_i^j}a)
		-
		\ln\gamma^i(\varpi^i_{I_i^j})
		\\
		&\hspace{4.5cm}
		-
		\ln\gamma^{0i}(\varpi^i_{I_i^j}a)
		+
		\ln\gamma^{0i}(\varpi^i_{I_i^j})
		\Bigr)
		\\
		\text{s.t.}
		&\quad
		\sum_{a\in A(I_i^j)}
		\gamma^i(\varpi^i_{I_i^j}a)
		-
		\gamma^i(\varpi^i_{I_i^j})
		=0,
		\quad j\in M_i.
	\end{aligned}
\end{equation}
The second term is a relative
dilated-entropy regularizer with respect
to $\bm{\gamma}^{0i}$, obtained by applying
the dilation operation at each
information set $I_i^j$ with
$\gamma^i(\varpi^i_{I_i^j})$
as the scaling variable
\cite{hiriarturrutyFundamentalsConvexAnalysis2001}.
Since dilated entropy is strictly convex
\cite{lingWhatGameAre2018},
Problem~\eqref{eq:sequence-form-best-response} has a strictly
concave objective, ensuring a unique
optimal response for fixed
$\bm{\hat{\gamma}^{-i}}$. The payoff term
incorporates the perturbation
$\bm{\pi_w}(\bm{\hat{\gamma}^{-i}},t)$,
which corresponds to
$\bm{\pi_s}(
\bm{\sigma}^{-i}(\bm{\hat{\gamma}^{-i}}),t)$
in mixed-strategy space.

A realization-plan profile $\bm{\gamma}^*$
is a Nash equilibrium of
$\Gamma_s^e(t)$ if and only if,
for every $i\in N$,
$\bm{\gamma}^{*i}$ solves
Problem~\eqref{eq:sequence-form-best-response}
with
$\bm{\hat{\gamma}^{-i}}=\bm{\gamma}^{*-i}$.
Together with
Lemma~\ref{lem:canonical-variant-strategy},
this yields the following equivalence.
\begin{theorem}\label{thm:positive-parameter-equivalence}
	Let $t\in(0,1]$ and
	$\bm{\gamma}^*\in\Lambda_{++}$.
	Suppose that the reference profile satisfies
	$\bm{\sigma^0}=\bm{\sigma}(\bm{\gamma^0})$,
	where $\bm{\gamma^0}\in\Lambda_{++}$.
	Then $\bm{\gamma}^*$ is a Nash equilibrium
	of $\Gamma_s^e(t)$ if and only if
	$\bm{\sigma}(\bm{\gamma}^*)$
	satisfies the logistic QRE variant
	in~\eqref{eq:logit-qre-variant}
	at rationality parameter $\lambda_p(t)$.
\end{theorem}

\begin{proof}
	Fix an arbitrary player $i\in N$
	and an opponents' realization-plan profile
	$\bm{\hat{\gamma}^{-i}}$. For any completely mixed strategy $\bm{\sigma^i}$
	and its induced realization plan $\bm{\gamma}^i$,
	the correspondence between mixed strategies
	and realization plans gives
	\[
	u^i\left(
	\bm{\sigma^i},
	\bm{\pi_s}\left(
	\bm{\sigma}^{-i}(\bm{\hat{\gamma}^{-i}}),t
	\right)
	\right)
	=g^i\left(
	\bm{\gamma}^i,
	\bm{\pi_w}(\bm{\hat{\gamma}^{-i}},t)
	\right).
	\]	
	By Lemma~\ref{lem:canonical-variant-strategy},
	the optimal response in
	Problem~\eqref{eq:normal-form-best-response}
	can be represented as
	$\bm{\sigma}^i(\bm{\gamma}^i)$. Thus, under the condition
	$\bm{\sigma}^{0i}=\bm{\sigma}^i(\bm{\gamma}^{0i})$,
	the strategy $\bm{\sigma}^i(\bm{\gamma}^i)$
	minimizes the relative-entropy term
	among all mixed strategies inducing
	$\bm{\gamma}^i$.
	Together with the payoff equivalence
	and the relative-entropy identity
	established above, this implies that
	Problems~\eqref{eq:normal-form-best-response}
	and~\eqref{eq:sequence-form-best-response}
	have equivalent best responses.
	Consequently, $\bm{\gamma}^*$ is a
	Nash equilibrium of $\Gamma_s^e(t)$
	if and only if
	$\bm{\sigma}(\bm{\gamma}^*)$
	is a Nash equilibrium of $\Gamma_p^e(t)$,
	which is equivalent to satisfying
	the logistic QRE variant
	in~\eqref{eq:logit-qre-variant}
	at rationality parameter $\lambda_p(t)$.
\end{proof}

Applying the first-order optimality
conditions to Problem~\eqref{eq:sequence-form-best-response}
for each player and imposing the
equilibrium condition
$\bm{\hat{\gamma}}=\bm{\gamma}$
yields the following system:
\begin{equation}\label{eq:sequence-form-equilibrium-system}
	\begin{aligned}
		&
		(1-c(t))
		g^i\bigl(
		\varpi^i_{I_i^j}a,
		\bm{\pi_w}(\bm{\gamma}^{-i},t)
		\bigr)-c(t)(1-m_i(\varpi^i_{I_i^j}a))
		\\
		&\hspace{1cm}-
		c(t)\bigl(
		\ln\gamma^i(\varpi^i_{I_i^j}a)
		-
		\ln\gamma^i(\varpi^i_{I_i^j})
		-
		\ln\gamma^{0i}(\varpi^i_{I_i^j}a)
		+
		\ln\gamma^{0i}(\varpi^i_{I_i^j})\bigr)
		\\
		&\hspace{4cm}
		-
		\nu^i_{I_i^j}
		+
		\zeta^i_{I_i^j}(a)
		=0,\;i\in N,
		j\in M_i,
		a\in A(I_i^j),
		\\
		&
		\sum_{a\in A(I_i^j)}
		\gamma^i(\varpi^i_{I_i^j}a)
		-
		\gamma^i(\varpi^i_{I_i^j})
		=0,\; j\in M_i,
		a\in A(I_i^j),
		\\
		&
		0<\gamma^i(\varpi^i_{I_i^j}a),
		\;
		i\in N,
		j\in M_i,
		a\in A(I_i^j),
	\end{aligned}
\end{equation}
where
$
\zeta^i_{I_i^j}(a)
=
\sum_{j_q\in M_i(\varpi^i_{I_i^j}a)}
\nu^i_{I_i^{j_q}}.
$
Since the objective of
Problem~\eqref{eq:sequence-form-best-response} is concave
in $\bm{\gamma}^i$ for fixed opponents,
the first-order optimality conditions
are necessary
and sufficient for global optimality
at an interior feasible point.
Together with
Theorem~\ref{thm:positive-parameter-equivalence},
this yields the following result.

\begin{theorem}\label{thm:sequence-system-characterization}
	Let $t\in(0,1]$ and suppose that
	$\bm{\sigma^0}
	=\bm{\sigma}(\bm{\gamma^0})$,
	where $\bm{\gamma^0}\in\Lambda_{++}$.
	Then 
	$\bm{\gamma}^*\in\Lambda_{++}$
	induces a logistic QRE variant
	$\bm{\sigma}(\bm{\gamma}^*)$
	satisfying~\eqref{eq:logit-qre-variant}
	at rationality parameter
	$\lambda_p(t)$ if and only if
	there exists $\bm{\nu}^*$ such that
	$(\bm{\gamma}^*,\bm{\nu}^*)$
	solves System~\eqref{eq:sequence-form-equilibrium-system}.
\end{theorem}

The preceding results establish a sequence-form characterization of the logistic QRE variant for each \(t\in(0,1]\). To establish its equilibrium-selection property, it remains to show that the corresponding solutions can be traced along a smooth path from an arbitrary interior realization plan at \(t=1\) and that their limit as \(t\to0\) is a normal-form perfect equilibrium.

\section{An Equilibrium Path from an Arbitrary Interior Realization Plan}\label{sec:equilibrium-path}

In this section, we investigate the equilibrium path associated with the logistic QRE variant in sequence form. We first establish the uniqueness of its initial solution at \(t=1\), followed by the existence of a connected component extending to \(t=0\). We then examine the smoothness of this component under a suitable perturbation and its limiting behavior in relation to normal-form perfect equilibrium.

\begin{lemma}\label{lem:unique-path-start}
	At $t=1$, System~\eqref{eq:sequence-form-equilibrium-system}
	admits a unique solution
	$(\bm{\gamma}^*(1),\bm{\nu}^*(1))$ satisfying $\gamma^{*i}(1;\varpi^i_{I_i^j}a)
	=
	\gamma^{0i}(\varpi^i_{I_i^j}a),
	\,
	\nu^{*i}_{I_i^j}(1)=-1$,
	for all $i\in N$, $j\in M_i$,
	and $a\in A(I_i^j)$.
\end{lemma}
\begin{proof}
	At $t=1$, we have $c(1)=1$ and
	$\lambda_p(1)=0$.
	Since $\bm{\sigma}^0$ is the unique logistic
	QRE variant at $\lambda_p(1)=0$,
	Theorem~\ref{thm:sequence-system-characterization} implies that
	any solution of System~\eqref{eq:sequence-form-equilibrium-system}
	satisfies $\bm{\gamma}^*(1)=\bm{\gamma}^0$. Substituting $\bm{\gamma}^*(1)=\bm{\gamma}^0$
	into the first group of equations
	in System~\eqref{eq:sequence-form-equilibrium-system} and using the definition of
	$\zeta^i_{I_i^j}(a)$ recursively
	gives
	$\nu^{*i}_{I_i^j}(1)=-1$
	for all $i\in N$ and $j\in M_i$. This completes the proof.
\end{proof}

We next establish the existence of a
connected component extending from
$t=1$ to $t=0$. The analysis relies
on the following fixed-point theorem
of Mas-Colell~\cite{mas-colellNoteTheoremBrowder1974}.
\begin{theorem}\label{thm:mas-colell-fixed-point}\textbf{(Mas-Colell's fixed point theorem).} Let $C$ be a nonempty, compact and convex subset of $\mathbb{R}^m$ and $h:C\times[0,1]\to C$ be an upper hemi-continuous mapping. Then the set $H =\{(z,t)\in C\times[0,1]\mid z=h(z,t)\}$ contains a connected set $H^c$ such that $C\times\{1\}\cap H^c\neq\emptyset$ and $C\times\{0\}\cap H^c\neq\emptyset$.
\end{theorem}
Let $\widetilde{\mathscr{S}}_D=\{(\bm{\gamma},\bm{\nu},t)\mid (\bm{\gamma},\bm{\nu},t) \text{ satisfies System~(\ref{eq:sequence-form-equilibrium-system}) with } 0<t\leq 1\}$ and $\mathscr{S}_D$ be the closure of $\widetilde{\mathscr{S}}_D$. By applying Theorem~\ref{thm:mas-colell-fixed-point}, we arrive at the following conclusion.

\begin{theorem}\label{thm:connected-solution-component}There is a connected component in $\mathscr{S}_D$ intersecting both $\mathbb{R}^{n_0}\times\mathbb{R}^{m_0}\times\{1\}$ and $\mathbb{R}^{n_0}\times\mathbb{R}^{m_0}\times\{0\}$.
\end{theorem}
\begin{proof}
	For each $(\bm{\hat\gamma},t)\in \Lambda\times[0,1]$, define the mapping $\varphi(\bm{\hat\gamma},t)$ as the set of realization-plan profiles $\bm{\gamma}=(\bm{\gamma}^i:i\in N)$ such that, for each player $i\in N$, $\bm{\gamma}^i$ solves Problem~\eqref{eq:sequence-form-best-response} when $t\in (0,1]$, and solves the following optimization problem when $t=0$
	\begin{align*}
		\max\limits_{\bm{\gamma}^i} & \quad\sum\limits_{j\in M_i}\sum\limits_{a\in A(I^j_i)}\gamma^i(\varpi^i_{I^j_i}a)g^i(\varpi^i_{I^j_i}a,\bm{\hat\gamma}^{-i})\\
		\text{s.t.} & \quad\sum\limits_{a\in A(I^j_i)}\gamma^i(\varpi^i_{I^j_i}a)-\gamma^i(\varpi^i_{I^j_i})=0,\;j\in M_i.
	\end{align*}
	By Theorem 2.2.2 of Fiacco~\cite{FiaccoIntroductionSensitivityStability1983}, $\varphi(\bm{\gamma}, t)$ is an upper hemi-continuous mapping from $\Lambda \times [0,1]$ to $\Lambda$. Let $\mathscr{E}=\{(\bm{\gamma},t)\in\Lambda\times[0,1]\mid\varphi(\bm{\gamma},t)=\bm{\gamma}\}$. Theorem~\ref{thm:mas-colell-fixed-point} then guarantees the existence of a connected component in $\mathscr{E}$ that intersects both $\mathbb{R}^{n_0}\times\{1\}$ and $\mathbb{R}^{n_0}\times\{0\}$. We denote this component by $\mathscr{E}^c$, and denote its restriction to $t>0$ by
	$\widetilde{\mathscr{E}}^c$.
	
	For any $(\bm{\gamma},t)\in \widetilde{\mathscr{E}}^c$, there exists a unique $\bm{\nu}=(\nu^i_{I^j_i}:i\in N,j\in M_i)$ such that System (\ref{eq:sequence-form-equilibrium-system}) is satisfied. Let $\widetilde{\mathscr{S}}^c_D=\{(\bm{\gamma},\bm{\nu},t)\in \widetilde{\mathscr{S}}_D\mid(\bm{\gamma},t)\in \widetilde{\mathscr{E}}^c\}$ and $\mathscr{S}^c_D$ be the closure of $\widetilde{\mathscr{S}}^c_D$. We obtain from the above discussion that $\mathscr{S}^c_D$ constitutes a connected component within $\mathscr{S}_D$ intersecting $\mathbb{R}^{n_0}\times\mathbb{R}^{m_0}\times\{1\}$. Considering a convergent sequence $\{(\bm{\gamma}(t_k), t_k)\}^\infty_{k=1} \subseteq \widetilde{\mathscr{E}}^c$ with $\lim_{k\to\infty} t_k = 0$, we associate each $(\bm{\gamma}(t_k), t_k)$ with the corresponding $\bm{\nu}(t_k)$, which is bounded as shown in Appendix~\ref{app:multiplier-boundedness}. The boundedness of $\{(\bm{\gamma}(t_k),\bm{\nu}(t_k),t_k)\}^\infty_{k=1}\subseteq \widetilde{\mathscr{S}}^c_D$ guarantees that it has a convergent subsequence. Thus, $\mathscr{S}^c_D$ intersects with $\mathbb{R}^{n_0}\times\mathbb{R}^{m_0}\times\{0\}$. This completes the proof.
\end{proof}

By Lemma~\ref{lem:unique-path-start}, the connected
component established in
Theorem~\ref{thm:connected-solution-component}
intersects the $t=1$ level at the unique
solution $(\bm{\gamma}^*(1),\bm{\nu}^*(1),1)$. In the normal-form formulation,
an auxiliary perturbation is introduced
to obtain a smooth solution path
by subtracting
$c(t)(1-t)\bar\alpha^i(s^i)$
from the first-order equation
corresponding to each $i\in N$
and $s^i\in S^i$
in System~\eqref{eq:normal-form-equilibrium-system},
following Cao et al.~
\cite{caoVariantLogisticQuantal2024}. To establish the corresponding
sequence-form perturbation,
let $\bm{\alpha}=(\alpha^i(\varpi^i_{I_i^j}a):i\in N,j\in M_i,a\in A(I^j_i))\in\mathbb{R}^{n_0}$
with sufficiently small $\|\bm{\alpha}\|$,
and satisfy, for each $i\in N$ and $s^i\in S^i$,
\begin{equation}\label{eq:auxiliary-perturbation-lift}
	\bar\alpha^i(s^i)
	=
	\sum_{j\in M_i}
	\sum_{a\in A(I_i^j)}
	s^i(\varpi^i_{I_i^j}a)
	\alpha^i(\varpi^i_{I_i^j}a).
\end{equation}
For any mixed strategy $\bm{\sigma}^i$
inducing the realization plan $\bm{\gamma}^i$,
we have
\begin{equation}\label{eq:auxiliary-perturbation-expectation}
	\sum_{s^i\in S^i}
	\sigma^i(s^i)\bar\alpha^i(s^i)
	=
	\sum_{j\in M_i}
	\sum_{a\in A(I_i^j)}
	\gamma^i(\varpi^i_{I_i^j}a)
	\alpha^i(\varpi^i_{I_i^j}a).
\end{equation}
Consequently, under the same reference
profile and strategy perturbation,
the corresponding sequence-form
auxiliary system is obtained
by subtracting $c(t)(1-t)\bm{\alpha}$
from the first group of equations
in System~\eqref{eq:sequence-form-equilibrium-system}.
Denote the resulting system by
$p(\bm{\gamma},\bm{\nu},t;\bm{\alpha})=0$ and write
$p_{\bm{\alpha}}(\bm{\gamma},\bm{\nu},t)
=p(\bm{\gamma},\bm{\nu},t;\bm{\alpha})$
for fixed $\bm{\alpha}$. Define $\widetilde{\mathscr{S}}_D(\bm{\alpha})
=
\left\{
(\bm{\gamma},\bm{\nu},t)
\;\middle|\;
p_{\bm{\alpha}}(\bm{\gamma},\bm{\nu},t)=0,\;\bm{\gamma}\in\Lambda_{++}\;
0<t\leq1
\right\}$,
and let $\mathscr{S}_D(\bm{\alpha})$
be its closure. We obtain the following result.

\begin{theorem}\label{thm:smooth-perfect-equilibrium-path}Given almost any $\bm{\alpha}\in\mathbb{R}^{n_0}$ with sufficiently small $\|\bm{\alpha}\|$, there exists a smooth path in $\mathscr{S}_D(\bm{\alpha})$ that starts from $(\bm{\gamma}^*(1),\bm{\nu}^*(1),1)$ on the level of $t=1$ and leads to a normal-form perfect equilibrium as $t\to 0$.
\end{theorem}
\begin{proof}
	With the continuous differentiability of $p(\bm{\gamma},\bm{\nu},t;\bm{\alpha})$ on $\text{int}(\Lambda)\times\mathbb{R}^{m_0}\times(0,1)\times\mathbb{R}^{n_0}$, we have proved in Appendix~\ref{app:jacobian-full-row-rank} that the Jacobian matrix of $p(\bm{\gamma},\bm{\nu},t;\bm{\alpha})$ is of full-row rank in this region. As an application of the transversality theorem outlined by Eaves and Schmedders~\cite{EavesGeneralequilibriummodels1999}, it can be shown that zero is a regular value of $p_{\bm{\alpha}}(\bm{\gamma},\bm{\nu},t)$ over $\text{int}(\Lambda)\times\mathbb{R}^{m_0}\times(0,1)$ for almost any $\bm{\alpha}$. 
	
	We fix $\bm{\alpha}$ such that zero is a regular value of $p_{\bm{\alpha}}(\bm{\gamma},\bm{\nu},t)$ over $\text{int}(\Lambda)\times\mathbb{R}^{m_0}\times(0,1)$. Since the perturbation
	$c(t)(1-t)\bm{\alpha}$
	preserves the feasibility constraints
	and vanishes at both endpoints,
	the argument of
	Theorem~\ref{thm:connected-solution-component}
	extends to $\mathscr{S}_D(\bm{\alpha})$.
	By the implicit function theorem,
	the resulting connected component
	defines a smooth path extending
	from the $t=1$ level to the
	$t=0$ level.
	Moreover, Appendix~\ref{app:jacobian-full-row-rank}
	establishes that zero remains a
	regular value of
	$p_{\bm{\alpha}}(\bm{\gamma},\bm{\nu},1)$
	with respect to $(\bm{\gamma},\bm{\nu})$.
	Hence, the path intersects the
	$t=1$ level transversely at the
	unique starting point
	$(\bm{\gamma}^*(1),\bm{\nu}^*(1),1)$.
	
	Finally, since both the regularization
	and the perturbation are of order
	$c(t)$, the corresponding
	best-response errors are also
	of order $c(t)$.
	Applying the error-bound analysis~\cite{caoVariantLogisticQuantal2024}	
	to the equilibrium system of the
	perturbed games shows that the
	distance to their equilibrium set
	decreases faster than $t$.
	Consequently, every limit point of
	$\bm{\sigma}(\bm{\gamma})$ along the smooth
	path as $t\to0$ is a normal-form
	perfect equilibrium of the
	original game.
\end{proof}
To reformulate System~\eqref{eq:sequence-form-equilibrium-system} without logarithmic functions, we use an exponential transformation on variables in the following. For $v\in\mathbb{R}$, let 
\[
\phi(v) =
\begin{cases}
	e^{1 - \frac{1}{v}}, & \text{if } v > 0, \\
	0, & \text{if } v \leq 0,
\end{cases}
\qquad\text{with}\quad
\frac{d}{dv} \phi(v) =
\begin{cases}
	\frac{e^{1 - \frac{1}{v}}}{v^2}, & \text{if } v > 0, \\
	0, & \text{if } v \leq 0.
\end{cases}
\]
Clearly, $\phi(v)$ is continuously differentiable on $\mathbb{R}$. Furthermore, $\phi(v)$ is a strictly increasing function on $[0, \infty)$ with $\phi(1) = 1$. Let $\bm{y}=(y^i(\varpi^i_{I^j_i}a):i\in N,j\in M_i,a\in A(I^j_i))\in \mathbb{R}^{n_0}$. We set $\gamma^i(\bm{y};\varpi^i_{I^j_i}a) = \phi(y^i(\varpi^i_{I^j_i}a)), i\in N,j\in M_i,a\in A(I^j_i)$. Substituting $\bm{\gamma}(\bm{y})=(\gamma^i(\bm{y};\varpi^i_{I^j_i}a):i\in N,j\in M_i,a\in A(I^j_i))$ into $p_{\bm{\alpha}}(\bm{\gamma},\bm{\nu},t)=0$ for $\bm{\gamma}$, we obtain
\begin{equation}\label{eq:sfde-exponential-system}\begin{aligned}
		& (1-c(t))g^i(\varpi^i_{I^j_i}a,\bm{\pi_w}(\bm{\gamma}^{-i}(\bm{y}),t))-c(t)(\ln\gamma^{0i}(\varpi^i_{I^j_i})-\ln\gamma^{0i}(\varpi^i_{I^j_i}a))\\
		& \hspace{1.2cm}+c(t)/y^i(\varpi^i_{I^j_i}a)-c(t)/y^i(\varpi^i_{I^j_i})-c(t)(1-m_i(\varpi^i_{I_i^j}a))-\nu^i_{I^j_i} + \zeta^i_{I^j_i}(a)\\
		& \hspace{3.2cm}-c(t)(1-t)\alpha(\varpi^i_{I^j_i}a) = 0,\;i\in N,j\in M_i,a\in A(I^j_i),\\
		& \sum\limits_{a\in A(I^j_i)}\gamma^i(\bm{y};\varpi^i_{I^j_i}a)-\gamma^i(\bm{y};\varpi^i_{I^j_i})=0,\;i\in N,j\in M_i,\\ & 0<y^i(\varpi^i_{I^j_i}a),\; i\in N,j\in M_i,a\in A(I^j_i),
	\end{aligned}
\end{equation}
where $y^i(\varpi^i_\emptyset) = 1$ for $i\in N$. We next introduce an additional variable transformation for handling fractional terms and inequalities in System~\eqref{eq:sfde-exponential-system}.

Given $\tau_0>0$ and $\kappa_0>2$, define two functions \[\psi_1(v,r;\tau_0,\kappa_0)=\left(\frac{v+\sqrt{v^2+4\tau_0r}}{2}\right)^{\kappa_0},\; \psi_2(v,r;\tau_0,\kappa_0)=\left(\frac{-v+\sqrt{v^2+4\tau_0r}}{2}\right)^{\kappa_0}.\] It follows that $\psi_1(v,r;\tau_0,\kappa_0)\psi_2(v,r;\tau_0,\kappa_0)=(\tau_0r)^{\kappa_0}$. Since $\kappa_0>2$, $\psi_1(v,r;\tau_0,\kappa_0)$ and $\psi_2(v,r;\tau_0,\kappa_0)$ are both continuously differentiable on $\mathbb{R}\times[0,\infty)$. Replacing $t/y^i(\varpi^i_{I^j_i}a)$ with $\xi^i(\varpi^i_{I^j_i}a)$ in System \eqref{eq:sfde-exponential-system}, we arrive at an equivalent system. Furthermore, for $\bm{x}=(x^i(\varpi^i_{I^j_i}a):i\in N,j\in M_i,a\in A(I^j_i))\in \mathbb{R}^{n_0}$, we define $\bm{y}(\bm{x},t)=(y^i(\bm{x},t;\varpi^i_{I^j_i}a):i\in N,j\in M_i,a\in A(I^j_i))$ and $\bm{\xi}(\bm{x},t)=(\xi^i(\bm{x},t;\varpi^i_{I^j_i}a):i\in N,j\in M_i,a\in A(I^j_i))$, where $y^i(\bm{x},t;\varpi^i_{I^j_i}a)=\psi_1(x^i(\varpi^i_{I^j_i}a),c(t)^{1/\kappa_0}; 1, \kappa_0)$ and $\xi^i(\bm{x},t;\varpi^i_{I^j_i}a)=\psi_2(x^i(\varpi^i_{I^j_i}a),c(t)^{1/\kappa_0};1,\kappa_0),\, i\in N,j\in M_i,a\in A(I^j_i)$. Through the utilization of $\bm{y}(\bm{x},t)$ and $\bm{\xi}(\bm{x},t)$, System~(\ref{eq:sfde-exponential-system}) can be equivalently reformulated as
\begin{equation}\label{eq:sfde-smooth-system}\begin{aligned}
		& (1-c(t))g^i(\varpi^i_{I^j_i}a,\bm{\pi_w}(\bm{\gamma}^{-i}(\bm{y}(\bm{x},t)),t))-c(t)(\ln\gamma^{0i}(\varpi^i_{I^j_i})-\ln\gamma^{0i}(\varpi^i_{I^j_i}a))\\
		& \hspace{1.4cm}+\xi^i(\bm{x},t;\varpi^i_{I^j_i}a)-\xi^i(\bm{x},t;\varpi^i_{I^j_i})-c(t)(1-m_i(\varpi^i_{I_i^j}a))-\nu^i_{I^j_i}+ \zeta^i_{I^j_i}(a)\\
		& \hspace{3.4cm}-c(t)(1-t)\alpha(\varpi^i_{I^j_i}a) = 0,\;i\in N,j\in M_i,a\in A(I^j_i),\\
		& \sum\limits_{a\in A(I^j_i)}\gamma^i(\bm{y}(\bm{x},t);\varpi^i_{I^j_i}a)-\gamma^i(\bm{y}(\bm{x},t);\varpi^i_{I^j_i})=0,\;i\in N,j\in M_i.
	\end{aligned}
\end{equation}
At $t=1$, the system admits a unique solution $(\bm{x}^*(1), \bm{\nu}^*(1))$ given by $x^{*i}(1;\varpi^i_{I^j_i}a) = (1-\ln\gamma^{0i}(\varpi^i_{I^j_i}a))^{-1/\kappa_0}-(1-\ln\gamma^{0i}(\varpi^i_{I^j_i}a))^{1/\kappa_0}$ for $i\in N,j\in M_i,a\in A(I^j_i)$, and $\nu^{*i}_{I^j_i}(1) = -1$ for $i\in N,j\in M_i$.

\begin{example}\label{ex:small-game-paths}
	We consider the two classical extensive-form games
	depicted in Figs.~\ref{fig:selten-game}
	and~\ref{fig:von-stengel-game}
	to illustrate the smooth solution paths generated
	by the proposed sequence-form formulation. For each game, we randomly choose a feasible interior
	realization plan $\bm{\gamma}^0\in\Lambda_{++}$. 
	Let $\bm{\alpha}=10^{-12}\bm{z}$,
	where $\bm{z}$ is randomly generated
	from $(0,1)^{n_0}$.
	Using the initial values specified above,
	we trace System~\eqref{eq:sfde-smooth-system}
	by the predictor--corrector method
	and recover the corresponding mixed-strategy
	path through~\eqref{eq:canonical-mixed-strategy}. The resulting paths are presented in
	Figs.~\ref{fig:selten-equilibrium-path}
	and~\ref{fig:von-stengel-equilibrium-path}, respectively.
	
	\begin{figure}[H]
		\centering
		\begin{minipage}[b]{0.47\textwidth}
			\centering
			\includegraphics[width=0.9\textwidth]{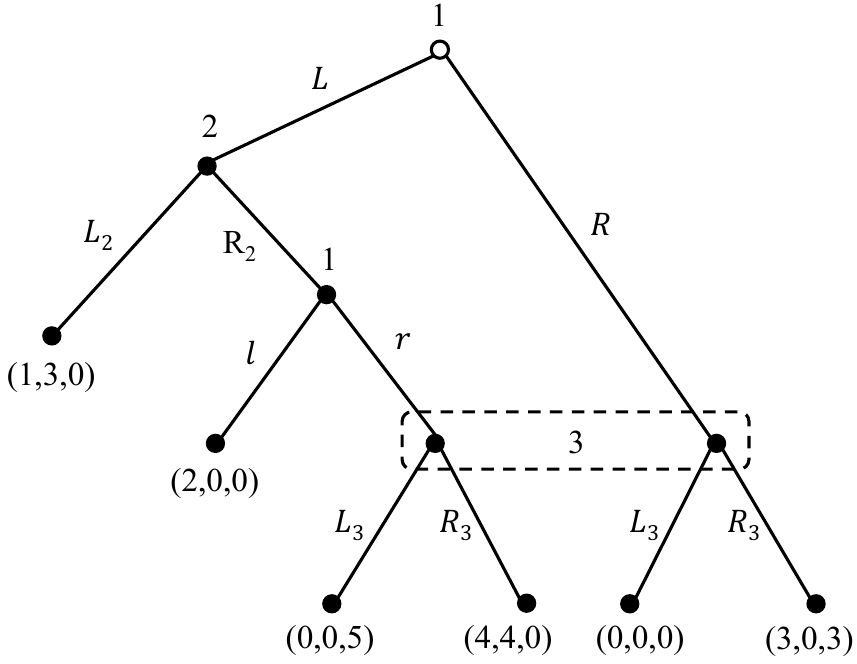}
			\caption{
				\label{fig:selten-game}
				An extensive-form game from
				Selten~\cite{SeltenReexaminationperfectnessconcept1975}
			}
		\end{minipage}\hfill
		\begin{minipage}[b]{0.52\textwidth}
			\centering
			\includegraphics[width=0.9\textwidth]{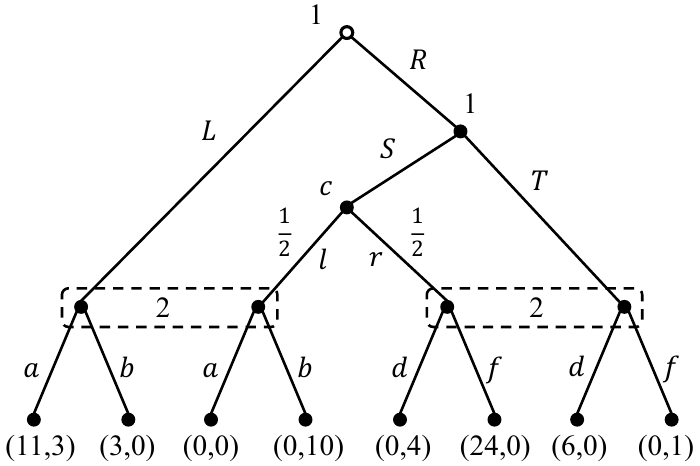}
			\caption{
				\label{fig:von-stengel-game}
				An extensive-form game from
				von Stengel et al.~\cite{vonStengelComputingNormalForm2002}
			}
		\end{minipage}
	\end{figure}
	
	To illustrate the numerical agreement between
	the solution paths generated by the proposed
	sequence-form and original normal-form formulations,
	we independently trace
	the corresponding normal-form auxiliary
	system based on~\eqref{eq:normal-form-equilibrium-system}, 
	using the initial mixed-strategy profile $\bm{\sigma}^0=\bm{\sigma}(\bm{\gamma}^0)$ and auxiliary perturbation
	specified by~\eqref{eq:auxiliary-perturbation-lift}.
	For each method, let
	$\{(t_k,\bm{\sigma}_k)\}_{k=0}^{K}$
	denote the computed continuation points
	and define
	\[
	d_k=\|\bm{\sigma}_k-\bm{\sigma}_K\|_2,
	\quad k=0,\ldots,K,
	\]
	where $\bm{\sigma}_K$ approximates
	the limiting equilibrium.
	Figs.~\ref{fig:selten-path-distance}
	and~\ref{fig:von-stengel-path-distance}
	compare the resulting distance profiles
	for the two games, respectively.
	
	\begin{figure}[H]
		\centering
		\begin{minipage}[b]{0.49\textwidth}
			\centering
			\includegraphics[
			height=0.25\textheight
			]{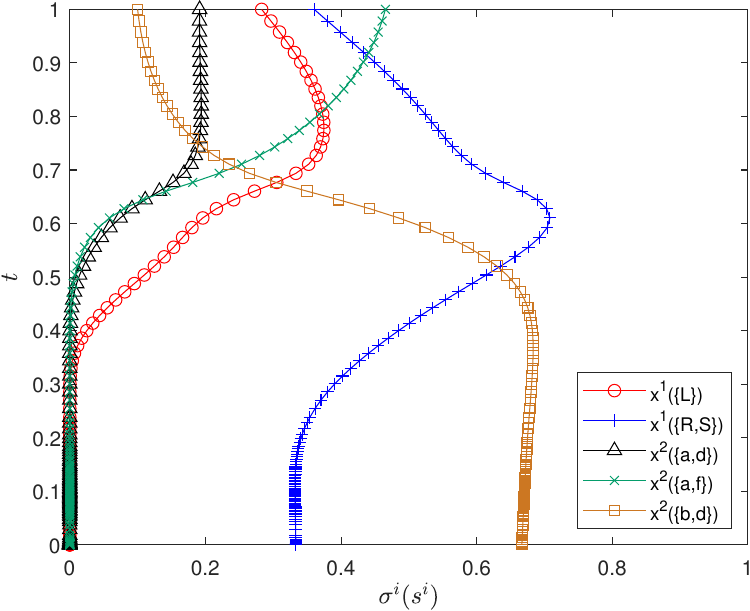}
			\caption{
				\label{fig:selten-equilibrium-path}
				Mixed-strategy paths for the game
				in Fig.~\ref{fig:selten-game}
			}
		\end{minipage}\hfill
		\begin{minipage}[b]{0.49\textwidth}
			\centering
			\includegraphics[
			height=0.25\textheight
			]{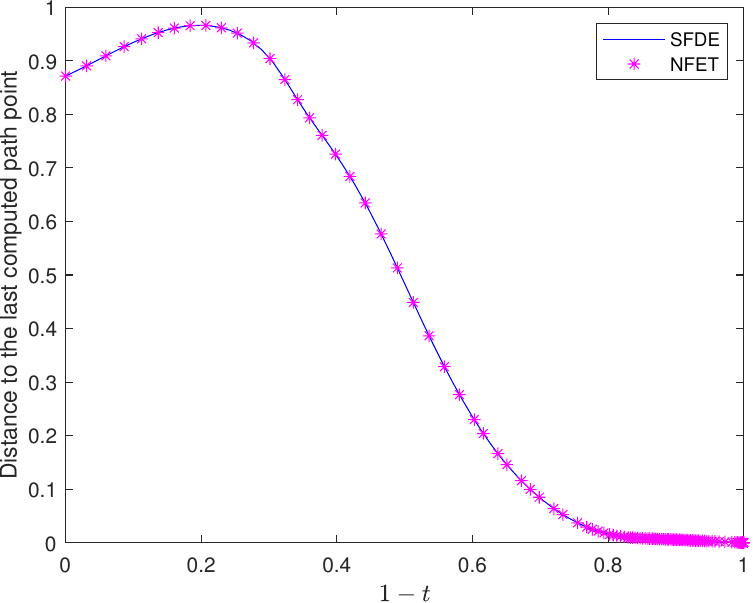}
			\caption{
				\label{fig:selten-path-distance}
				Distances to the terminal equilibrium
				for the game in Fig.~\ref{fig:selten-game}
			}
		\end{minipage}
	\end{figure}
	
	\begin{figure}[H]
		\centering
		\begin{minipage}[b]{0.49\textwidth}
			\centering
			\includegraphics[
			height=0.25\textheight
			]{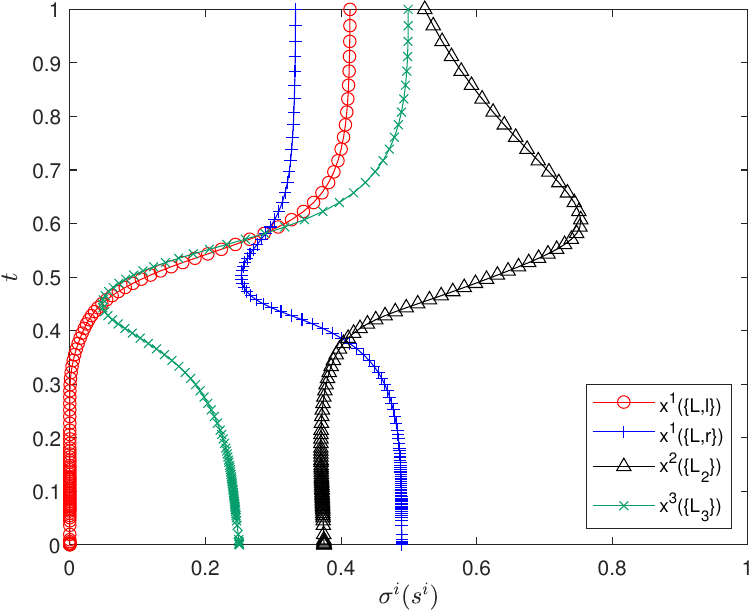}
			\caption{
				\label{fig:von-stengel-equilibrium-path}
				Mixed-strategy paths for the game
				in Fig.~\ref{fig:von-stengel-game}
			}
		\end{minipage}\hfill
		\begin{minipage}[b]{0.49\textwidth}
			\centering
			\includegraphics[
			height=0.25\textheight
			]{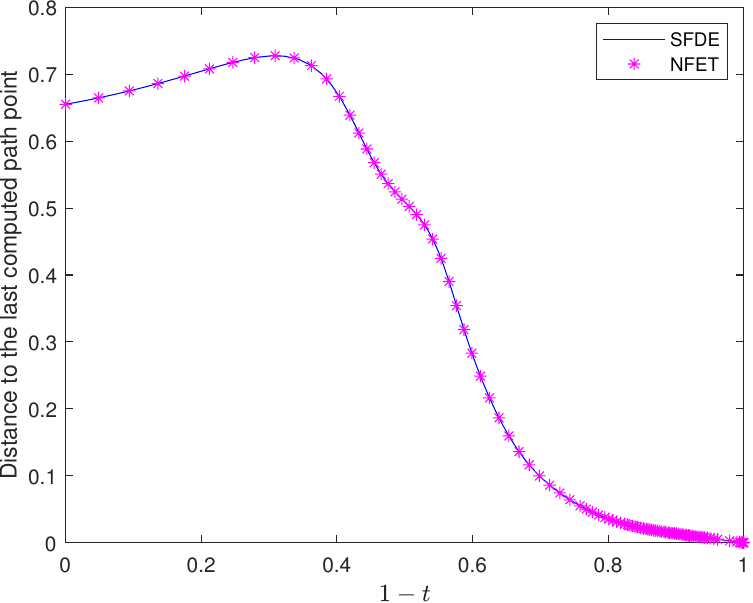}
			\caption{
				\label{fig:von-stengel-path-distance}
				Distances to the terminal equilibrium
				for the game in Fig.~\ref{fig:von-stengel-game}
			}
		\end{minipage}
	\end{figure}
	
\end{example}
\section{An Equivalent Smooth Path via Entropy Reformulation}\label{sec:equivalent-smooth-path}

By the realization-plan constraints~\eqref{eq:realization-flow-constraints},
for each $i\in N$ and $j\in M_i$, we have
$\sum_{a \in A(I_i^j)} \gamma^i(\varpi^i_{I_i^j} a) \big( \ln \gamma^i(\varpi^i_{I_i^j})- \ln \gamma^{0i}(\varpi^i_{I_i^j})\big)=\gamma^i(\varpi^i_{I_i^j} ) \big( \ln \gamma^i(\varpi^i_{I_i^j})- \ln \gamma^{0i}(\varpi^i_{I_i^j})\big)$. Consequently, the relative dilated-entropy
regularizer in Problem~\eqref{eq:sequence-form-best-response}
can be expressed equivalently as a
weighted sum of standard entropy terms.
This yields the following equivalent
optimization problem:
\begin{equation}\label{eq:weighted-entropy-best-response}
	\begin{aligned}
		\max_{\bm{\gamma}^i}
		&\quad
		(1-c(t))
		\sum_{j\in M_i}
		\sum_{a\in A(I_i^j)}
		\gamma^i(\varpi^i_{I_i^j}a)
		g^i\bigl(
		\varpi^i_{I_i^j}a,
		\bm{\pi_w}(\bm{\hat{\gamma}}^{-i},t)
		\bigr)
		\\
		&\quad
		-c(t)
		\sum_{j\in M_i}
		\sum_{a\in A(I_i^j)}
		\bigl(
		1-m_i(\varpi^i_{I_i^j}a)
		\bigr)
		\gamma^i(\varpi^i_{I_i^j}a)
		\Bigl(
		\ln\gamma^i(\varpi^i_{I_i^j}a)
		-
		\ln\gamma^{0i}(\varpi^i_{I_i^j}a)
		\Bigr)
		\\
		\text{s.t.}
		&\quad
		\sum_{a\in A(I_i^j)}
		\gamma^i(\varpi^i_{I_i^j}a)
		-
		\gamma^i(\varpi^i_{I_i^j})
		=0,
		\quad j\in M_i.
	\end{aligned}
\end{equation}
Although some coefficients
$1-m_i(\varpi^i_{I_i^j}a)$
may be negative, the objective of
Problem~\eqref{eq:weighted-entropy-best-response} remains
concave on the feasible realization-plan
set, as it coincides with that of
Problem~\eqref{eq:sequence-form-best-response} under the
realization-plan constraints.
Thus, the first-order optimality
conditions are necessary and sufficient
for global optimality at an interior
feasible point.
Imposing the equilibrium consistency
condition
$\bm{\hat{\gamma}}=\bm{\gamma}$
yields the following system:
\begin{equation}\label{eq:weighted-entropy-equilibrium-system}
	\begin{aligned}
		&
		(1-c(t))
		g^i\bigl(
		\varpi^i_{I_i^j}a,
		\bm{\pi_w}(\bm{\gamma}^{-i},t)
		\bigr)-
		c(t)
		\bigl(
		1-m_i(\varpi^i_{I_i^j}a)
		\bigr)
		\Bigl(
		\ln\gamma^i(\varpi^i_{I_i^j}a)
		-
		\ln\gamma^{0i}(\varpi^i_{I_i^j}a)
		\Bigr)
		\\
		&\hspace{2cm}
		-
		c(t)
		\bigl(
		1-m_i(\varpi^i_{I_i^j}a)
		\bigr)-
		\nu^i_{I_i^j}
		+
		\zeta^i_{I_i^j}(a)
		=0,
		\;
		i\in N,
		j\in M_i,
		a\in A(I_i^j),
		\\
		&
		\sum_{a\in A(I_i^j)}
		\gamma^i(\varpi^i_{I_i^j}a)
		-
		\gamma^i(\varpi^i_{I_i^j})
		=0,
		\;
		i\in N,
		j\in M_i,
		\\
		&
		\gamma^i(\varpi^i_{I_i^j}a)>0,
		\;
		i\in N,
		j\in M_i,
		a\in A(I_i^j),
	\end{aligned}
\end{equation}
where $\zeta^i_{I_i^j}(a)
=
\sum_{j_q\in M_i(\varpi^i_{I_i^j}a)}
\nu^i_{I_i^{j_q}}$. Since Problems~\eqref{eq:sequence-form-best-response}
and~\eqref{eq:weighted-entropy-best-response}
have identical objectives on the
feasible realization-plan set,
Systems~\eqref{eq:sequence-form-equilibrium-system}
and~\eqref{eq:weighted-entropy-equilibrium-system}
have the same realization-plan
solutions, although their associated
multipliers may differ.
By Theorem~\ref{thm:sequence-system-characterization},
each such solution induces a
logistic QRE variant satisfying
\eqref{eq:logit-qre-variant}
at rationality parameter
$\lambda_p(t)$.

To construct an equivalent smooth path,
we apply the same perturbation
$c(t)(1-t)\bm{\alpha}$ to
System~\eqref{eq:weighted-entropy-equilibrium-system},
subtracting it from the left-hand
side of the first group of equations.
Let $\widetilde{\mathscr{S}}_E(\bm{\alpha})$
denote the set of all triples
$(\bm{\gamma},\bm{\nu},t)$ satisfying the
resulting perturbed equations,
the realization-plan feasibility
constraints, and the positivity
conditions for $0<t\leq1$,
and let $\mathscr{S}_E(\bm{\alpha})$
be its closure.
The following theorem establishes
the equivalence of the two smooth paths.

\begin{theorem}\label{thm:equivalent-sequence-paths}
	For any fixed
	$\bm{\alpha}\in\mathbb{R}^{n_0}$,
	the perturbed systems defining
	$\widetilde{\mathscr{S}}_D(\bm{\alpha})$
	and
	$\widetilde{\mathscr{S}}_E(\bm{\alpha})$
	have identical realization-plan
	solutions for every $t\in(0,1]$.
\end{theorem}
\begin{proof}
	By the realization-plan constraints,
	Problems~\eqref{eq:sequence-form-best-response}
	and~\eqref{eq:weighted-entropy-best-response}
	have identical objectives
	on the feasible realization-plan set.
	The perturbation $c(t)(1-t)\bm{\alpha}$
	introduces the same linear term
	into both objectives and thus
	preserves their equivalence.
	Since the objectives are concave
	on the feasible realization-plan set,
	their first-order optimality conditions
	are necessary and sufficient
	for interior optimal responses.
	Consequently, the two perturbed
	systems have identical realization-plan
	solutions for every $t\in(0,1]$.
\end{proof}
By Theorems~\ref{thm:smooth-perfect-equilibrium-path}
and~\ref{thm:equivalent-sequence-paths},
the smooth path in
$\mathscr{S}_D(\bm{\alpha})$
admits an equivalent representation
in $\mathscr{S}_E(\bm{\alpha})$
with identical realization-plan profiles.
The corresponding multipliers are
related by a smooth transformation,
which reduces to the identity at $t=1$
since $\bm{\gamma}^*(1)=\bm{\gamma}^0$.
Thus, both paths originate at
$(\bm{\gamma}^*(1),\bm{\nu}^*(1),1)$
and have the same limiting
mixed-strategy profiles as $t\to0$.

We next reformulate the perturbed version of
System~\eqref{eq:weighted-entropy-equilibrium-system} to eliminate the logarithmic terms,
following the approach developed in the preceding section.
For each $i\in N$, define
\[
Q_i=
\left\{
(j,a)\mid
j\in M_i,\;
a\in A(I_i^j),\;
m_i(\varpi^i_{I_i^j}a)\neq 1
\right\}.
\]
Let
\[
\gamma^i(\bm{y};\varpi^i_{I_i^j}a)=
\begin{cases}
	\phi\bigl(y^i(\varpi^i_{I_i^j}a)\bigr),
	& (j,a)\in Q_i,\\[2mm]
	y^i(\varpi^i_{I_i^j}a),
	& (j,a)\notin Q_i.
\end{cases}
\]
Substituting $\bm{\gamma}(\bm{y})$ for $\bm{\gamma}$ in the perturbed
version of System~\eqref{eq:weighted-entropy-equilibrium-system}, we obtain
\begin{equation}\label{eq:sfet-exponential-system}
	\begin{aligned}
		&(1-c(t))
		g^i\bigl(
		\varpi^i_{I_i^j}a,
		\bm{\pi_w}(\bm{\gamma}^{-i}(\bm{y}),t)
		\bigr)\\
		&\quad
		-c(t)
		\bigl(1-m_i(\varpi^i_{I_i^j}a)\bigr)
		\left(
		2-\ln\gamma^{0i}(\varpi^i_{I_i^j}a)
		-\frac{1}{y^i(\varpi^i_{I_i^j}a)}
		\right)\\
		&\hspace{1.5cm}
		-\nu^i_{I_i^j}
		+\zeta^i_{I_i^j}(a)
		-c(t)(1-t)\alpha(\varpi^i_{I_i^j}a)
		=0,
		\; i\in N,(j,a)\in Q_i,\\
		&(1-c(t))
		g^i\bigl(
		\varpi^i_{I_i^j}a,
		\bm{\pi_w}(\bm{\gamma}^{-i}(\bm{y}),t)
		\bigr)
		-\nu^i_{I_i^j}
		+\zeta^i_{I_i^j}(a)\\
		&\hspace{2cm}
		-c(t)(1-t)\alpha(\varpi^i_{I_i^j}a)
		=0,
		\; i\in N,(j,a)\notin Q_i,\\
		&\sum_{a\in A(I_i^j)}
		\gamma^i(\bm{y};\varpi^i_{I_i^j}a)
		-\gamma^i(\bm{y};\varpi^i_{I_i^j})
		=0,
		\; i\in N,j\in M_i,
	\end{aligned}
\end{equation}
where
$y^i(\varpi^i_{I_i^j}a)>0$
for $i\in N$ and $(j,a)\in Q_i$. To remove the fractional terms, we introduce
$\xi^i(\varpi^i_{I_i^j}a)
=c(t)/y^i(\varpi^i_{I_i^j}a)$
for $i\in N$ and $(j,a)\in Q_i$.
For $\bm{x}\in\mathbb{R}^{n_0}$, define
\[
y^i(\bm{x},t;\varpi^i_{I_i^j}a)=
\begin{cases}
	\psi_1\bigl(
	x^i(\varpi^i_{I_i^j}a),
	c(t)^{1/\kappa_0};
	1,\kappa_0
	\bigr),
	& (j,a)\in Q_i,\\[2mm]
	x^i(\varpi^i_{I_i^j}a),
	& (j,a)\notin Q_i,
\end{cases}
\]
and, for $(j,a)\in Q_i$,
\[
\xi^i(\bm{x},t;\varpi^i_{I_i^j}a)
=
\psi_2\bigl(
x^i(\varpi^i_{I_i^j}a),
c(t)^{1/\kappa_0};
1,\kappa_0
\bigr).
\]
By the definitions of $\psi_1$ and $\psi_2$, $y^i(\bm{x},t;\varpi^i_{I_i^j}a)
\xi^i(\bm{x},t;\varpi^i_{I_i^j}a)
=c(t),\, i\in N,(j,a)\in Q_i$.
Substituting these transformations into
the perturbed
version of System~\eqref{eq:weighted-entropy-equilibrium-system}
yields
\begin{equation}\label{eq:sfet-smooth-system}
	\begin{aligned}
		&(1-c(t))
		g^i\bigl(
		\varpi^i_{I_i^j}a,
		\bm{\pi_w}(\bm{\gamma}^{-i}(\bm{y}(\bm{x},t)),t)
		\bigr)\\
		&\quad
		-c(t)
		\bigl(1-m_i(\varpi^i_{I_i^j}a)\bigr)
		\bigl(
		2-\ln\gamma^{0i}(\varpi^i_{I_i^j}a)
		\bigr)
		+\bigl(1-m_i(\varpi^i_{I_i^j}a)\bigr)
		\xi^i(\bm{x},t;\varpi^i_{I_i^j}a)
		\\
		&\hspace{2cm}
		-\nu^i_{I_i^j}
		+\zeta^i_{I_i^j}(a)-c(t)(1-t)\alpha(\varpi^i_{I_i^j}a)
		=0,
		\; i\in N,(j,a)\in Q_i,\\[1mm]
		&(1-c(t))
		g^i\bigl(
		\varpi^i_{I_i^j}a,
		\bm{\pi_w}(\bm{\gamma}^{-i}(\bm{y}(\bm{x},t)),t)
		\bigr)
		-\nu^i_{I_i^j}
		+\zeta^i_{I_i^j}(a)\\
		&\hspace{2cm}
		-c(t)(1-t)\alpha(\varpi^i_{I_i^j}a)
		=0,
		\; i\in N,(j,a)\notin Q_i,\\[1mm]
		&\sum_{a\in A(I_i^j)}
		\gamma^i(\bm{y}(\bm{x},t);\varpi^i_{I_i^j}a)
		-\gamma^i(\bm{y}(\bm{x},t);\varpi^i_{I_i^j})
		=0,
		\; i\in N,j\in M_i.
	\end{aligned}
\end{equation}

For $0<t\leq1$, System~\eqref{eq:sfet-smooth-system}
has the same realization-plan solutions
as the perturbed
version of System~\eqref{eq:weighted-entropy-equilibrium-system}.
At $t=1$, its unique solution
$(\bm{x}^*(1),\bm{\nu}^*(1))$ satisfies
\[
x^{*i}(1;\varpi^i_{I_i^j}a)=
\begin{cases}
	\bigl(
	1-\ln\gamma^{0i}(\varpi^i_{I_i^j}a)
	\bigr)^{-1/\kappa_0}
	-
	\bigl(
	1-\ln\gamma^{0i}(\varpi^i_{I_i^j}a)
	\bigr)^{1/\kappa_0},
	& (j,a)\in Q_i,\\[2mm]
	\gamma^{0i}(\varpi^i_{I_i^j}a),
	& (j,a)\notin Q_i,
\end{cases}
\]
with
$\nu^{*i}_{I_i^j}(1)=-1$
for $i\in N$ and $j\in M_i$.

Since the logistic-QRE variant reduces
to standard logistic QRE
as $\varepsilon_0\to0^+$
under a uniform reference profile,
the two sequence-form systems
developed above also encompass
the standard case.
Consequently, the associated
path-following methods provide
a computational approach to
standard logistic QRE in
finite $n$-player extensive-form games
without explicitly constructing
the reduced normal form.

\section{Numerical Experiments}\label{sec:numerical-experiments}

We evaluate the computational performance of
the two proposed sequence-form path-following
methods and compare them with the original
normal-form method.
The two sequence-form methods, denoted by
SFDE and SFET, are based on
Systems~\eqref{eq:sfde-smooth-system}
and~\eqref{eq:sfet-smooth-system}, respectively,
while the normal-form method, denoted by NFET,
is based on System~\eqref{eq:normal-form-equilibrium-system}.
As established in Theorem~\ref{thm:equivalent-sequence-paths},
SFDE and SFET yield identical realization-plan
solutions for the same perturbation.
Nevertheless, their different nonlinear
representations may lead to different
numerical performance.
The comparison with NFET further allows us
to assess the computational benefits
of the sequence-form formulation.
All three methods employ the predictor--corrector
continuation procedure to trace their respective
solution paths numerically.
For detailed accounts of predictor--corrector
methods, see Allgower and Georg~
\cite{AllgowerNumericalContinuationMethods1990}.

We consider two types of randomly generated
extensive-form games, illustrated in
Figs.~\ref{fig:type-1-game} and~\ref{fig:type-2-game},
following the numerical experiments
in~\cite{houSequenceFormCharacterizationDifferentiable2025}.
Each game is characterized by
$(n,\mathcal{L},\mathcal{A})$,
where $n$ denotes the number of players,
$\mathcal{L}$ is the maximum history depth,
and $\mathcal{A}$ is the number of actions
available at each information set.
Players move in cyclic order,
and their payoffs at each terminal node
are independently drawn from the uniform
distribution on $[-10,10]$.
We fix $n=3$ for Type-1 games and $n=4$
for Type-2 games, while varying
$\mathcal{L}$ and $\mathcal{A}$
to generate games of different sizes.

\begin{figure}[H]
	\centering
	\begin{minipage}[b]{0.55\textwidth}
		\centering
		\includegraphics[width=0.95\textwidth]{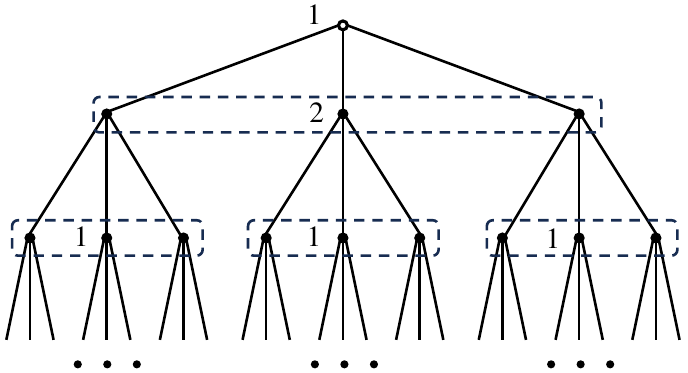}
		\caption{
			\label{fig:type-1-game}
			Structure of a Type-1 Game
		}
	\end{minipage}\hfill
	\begin{minipage}[b]{0.43\textwidth}
		\centering
		\includegraphics[width=0.95\textwidth]{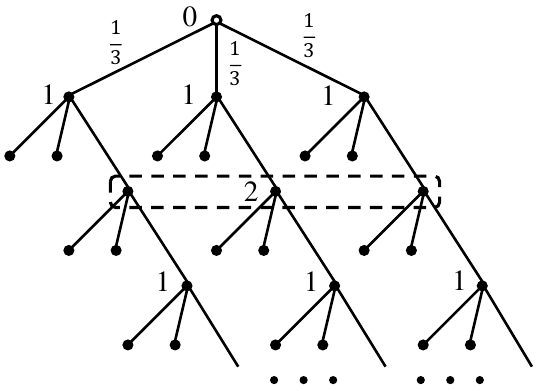}
		\caption{
			\label{fig:type-2-game}
			Structure of a Type-2 Game
		}
	\end{minipage}
\end{figure}

For each game type and parameter configuration
$(\mathcal{L},\mathcal{A})$,
we generate 20 instances with independently
sampled payoff specifications. For each instance, we randomly choose
an interior realization plan
$\bm{\gamma}^0\in\Lambda_{++}$ and construct the corresponding
completely mixed reference profile
$\bm{\sigma}^0=\bm{\sigma}(\bm{\gamma}^0)$
using~\eqref{eq:canonical-mixed-strategy}. For SFDE and SFET,
we use the same auxiliary perturbation
$\bm{\alpha}=10^{-12}\bm{z}$,
where $\bm{z}$ is randomly generated
from $(0,1)^{n_0}$.
For NFET, we use the corresponding
normal-form auxiliary perturbation
$\bm{\bar\alpha}$ determined by~\eqref{eq:auxiliary-perturbation-lift}. Thus, all three methods
trace equivalent solution paths.
The predictor step length is initialized
at $t^{0.2}$ and adjusted until the
predicted residual is below $0.1t^{0.5}$,
while the corrector tolerance is set
to $10^{-16}t^{0.02}$.
These numerical settings are applied
consistently to all three methods.
A run is considered successful when
$t<10^{-4}$ and unsuccessful when
either the prescribed iteration limit
or the computational-time limit is exceeded.
All computations are conducted on
a Windows Server 2016 Standard system
equipped with two Intel Xeon E5-2650 v4
processors operating at 2.20 GHz
and 128 GB of RAM.

We first compare the numbers of variables
in SFDE, SFET, and NFET under different
game configurations.
For a fixed value of the path parameter $t$,
each formulation yields a square system
of nonlinear equations.
Since $t$ is excluded from the variable counts,
the numbers of equations coincide with
the numbers of unknowns and are
therefore not reported separately.
Moreover, SFDE and SFET have the same
number of variables, so their counts
are represented by a single curve.
Fig.~\ref{fig:variable-count-comparison}
presents the comparisons for
Type-1 and Type-2 games.
In each panel, the first segment varies
$\mathcal{L}$ with $\mathcal{A}=2$ fixed,
while the second varies $\mathcal{A}$
with $\mathcal{L}=4$ for Type-1 games
and $\mathcal{L}=10$ for Type-2 games.
As either parameter increases,
the number of variables in NFET
grows substantially faster than
those in SFDE and SFET.
Except for the Type-1 configuration $(5,2)$,
the sequence-form methods require
fewer variables across all configurations,
with the difference becoming
increasingly pronounced as the game size grows.
These results demonstrate the substantial
reduction in system dimension achieved
by the sequence-form representation.

\begin{figure}[htbp]
	\centering
	\begin{subfigure}[b]{0.48\textwidth}
		\centering
		\includegraphics[width=\linewidth]{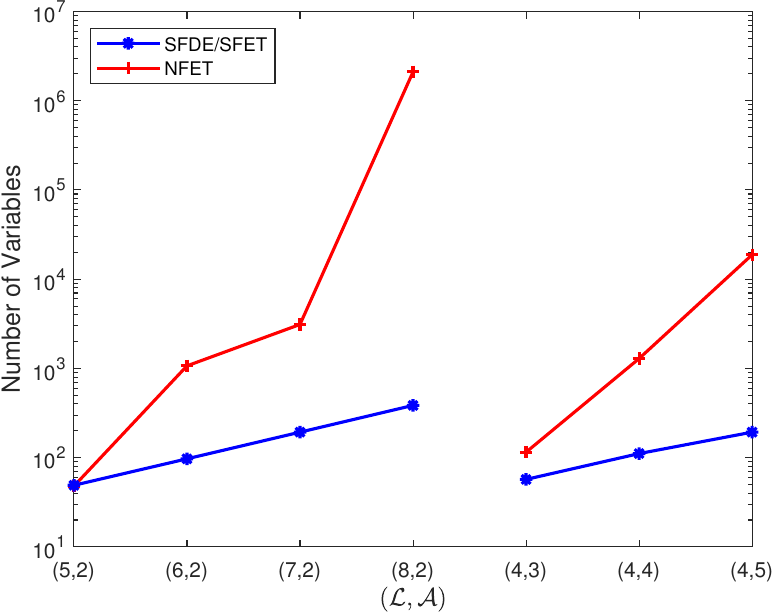}
		\caption{Type-1 Games}
		\label{fig:type-1-variable-count}
	\end{subfigure}
	\hfill
	\begin{subfigure}[b]{0.48\textwidth}
		\centering
		\includegraphics[width=\linewidth]{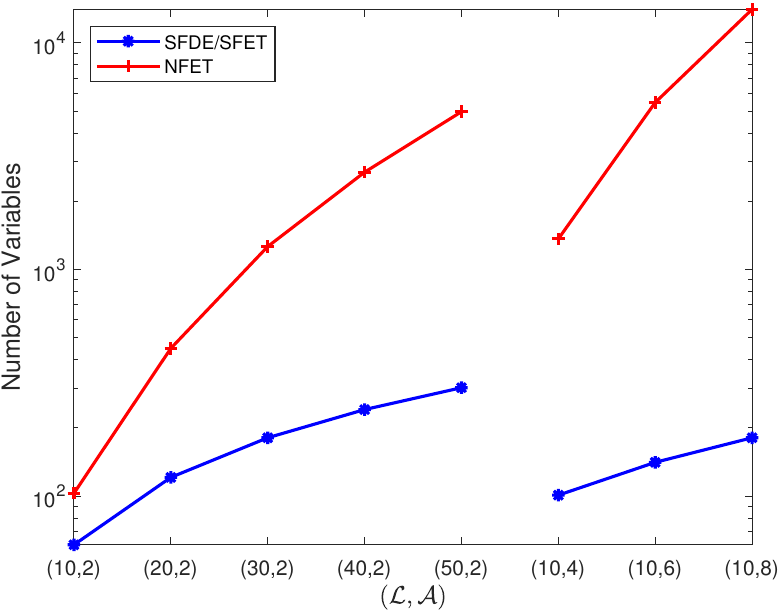}
		\caption{Type-2 Games}
		\label{fig:type-2-variable-count}
	\end{subfigure}
	
	\caption{
		Comparison of Variable Counts
		for SFDE, SFET, and NFET. The two line segments correspond
		to variations in $\mathcal{L}$
		and $\mathcal{A}$, respectively.
	}
	\label{fig:variable-count-comparison}
\end{figure}

We further evaluate the computational
performance of SFDE, SFET, and NFET
in terms of iteration count,
computational time, and failure rate.
For each parameter configuration,
we report the maximum, minimum,
and median iteration counts and
computational times, together with
the failure rates across all 20 instances.
A dash indicates that the corresponding
statistic is unavailable due to failed runs;
it does not necessarily imply that
all runs have failed.
Tables~\ref{tab:type-1-computational-results} and~\ref{tab:type-2-computational-results}
present the results for Type-1
and Type-2 games, respectively.
For Type-2 games, NFET fails on
all 20 instances in every configuration.
Its iteration counts and computational
times are therefore omitted from
Table~\ref{tab:type-2-computational-results}, while its failure
rates are retained for comparison.

\begin{table}[htbp]
	\renewcommand\arraystretch{0.95}
	\setlength{\tabcolsep}{3.4pt}
	\caption{
		Computational Performance of
		SFDE, SFET, and NFET for Type-1 Games
	}\small
	\label{tab:type-1-computational-results}
	\begin{tabular*}{\textwidth}{
			@{\extracolsep\fill}
			>{\rowmac}c
			>{\rowmac}l
			*{8}{>{\rowmac}c}
			>{\rowmac}c<{\clearrow}
		}
		\toprule
		\multirow{2}*{$(\mathcal{L},\mathcal{A})$}
		& & \multicolumn{3}{c}{Iteration Count}
		& \multicolumn{3}{c}{Computational Time (s)}
		& \multicolumn{3}{c}{Failure Rate}\\
		\cmidrule(r){3-5}
		\cmidrule(r){6-8}
		\cmidrule(r){9-11}
		& & SFDE & SFET & NFET
		& SFDE & SFET & NFET
		& SFDE & SFET & NFET\\
		\midrule
		
		\multirow{3}*{$(5,2)$}
		& max & - & - & -
		& - & - & -
		& \multirow{3}*{5\%}
		& \multirow{3}*{5\%}
		& \multirow{3}*{5\%}\\
		& min & 166 & 163 & 281
		& 6.9 & 6.5 & 216.6
		& & &\\
		&\setrow{\bfseries} med
		& 194.0 & 206.5 & 314.0
		& 8.1 & 8.0 & 340.1
		& & &\\
		
		\multirow{3}*{$(6,2)$}
		& max & - & 1283 & -
		& - & 364.3 & -
		& \multirow{3}*{5\%}
		& \multirow{3}*{0\%}
		& \multirow{3}*{100\%}\\
		& min & 191 & 172 & -
		& 30.9 & 26.3 & -
		& & &\\
		&\setrow{\bfseries} med
		& 209.5 & 255.0 & -
		& 34.1 & 39.4 & -
		& & &\\
		
		\multirow{3}*{$(7,2)$}
		& max & - & - & -
		& - & - & -
		& \multirow{3}*{10\%}
		& \multirow{3}*{30\%}
		& \multirow{3}*{100\%}\\
		& min & 206 & 253 & -
		& 173.1 & 182.1 & -
		& & &\\
		&\setrow{\bfseries} med
		& 256.0 & 549.5 & -
		& 226.5 & 574.1 & -
		& & &\\
		
		\multirow{3}*{$(8,2)$}
		& max & - & - & -
		& - & - & -
		& \multirow{3}*{45\%}
		& \multirow{3}*{95\%}
		& \multirow{3}*{100\%}\\
		& min & 221 & 338 & -
		& 1174.3 & 1608.1 & -
		& & &\\
		&\setrow{\bfseries} med
		& 305.5 & - & -
		& 1753.0 & - & -
		& & &\\
		
		\multirow{3}*{$(4,3)$}
		& max & 24138 & 2440 & -
		& 1800.1 & 980.2 & -
		& \multirow{3}*{0\%}
		& \multirow{3}*{0\%}
		& \multirow{3}*{100\%}\\
		& min & 188 & 196 & -
		& 19.5 & 19.6 & -
		& & &\\
		&\setrow{\bfseries} med
		& 223.5 & 231.5 & -
		& 22.9 & 23.8 & -
		& & &\\
		
		\multirow{3}*{$(4,4)$}
		& max & - & - & -
		& - & - & -
		& \multirow{3}*{15\%}
		& \multirow{3}*{35\%}
		& \multirow{3}*{100\%}\\
		& min & 222 & 257 & -
		& 153.2 & 176.3 & -
		& & &\\
		&\setrow{\bfseries} med
		& 266.5 & 406.0 & -
		& 195.1 & 340.5 & -
		& & &\\
		
		\multirow{3}*{$(4,5)$}
		& max & - & - & -
		& - & - & -
		& \multirow{3}*{70\%}
		& \multirow{3}*{80\%}
		& \multirow{3}*{100\%}\\
		& min & 272 & 309 & -
		& 992.1 & 1136.2 & -
		& & &\\
		&\setrow{\bfseries} med
		& - & - & -
		& - & - & -
		& & &\\
		
		\bottomrule
	\end{tabular*}
\end{table}

\begin{table}[htbp]
	\renewcommand\arraystretch{0.95}
	\setlength{\tabcolsep}{3.4pt}
	\caption{
		Computational Performance of
		SFDE, SFET, and NFET for Type-2 Games
	}\small
	\label{tab:type-2-computational-results}
	\begin{tabular*}{\textwidth}{
			@{\extracolsep\fill}
			>{\rowmac}c
			>{\rowmac}l
			*{6}{>{\rowmac}c}
			>{\rowmac}c<{\clearrow}
		}
		\toprule
		\multirow{2}*{$(\mathcal{L},\mathcal{A})$}
		& & \multicolumn{2}{c}{Iteration Count}
		& \multicolumn{2}{c}{Computational Time (s)}
		& \multicolumn{3}{c}{Failure Rate}\\
		\cmidrule(r){3-4}
		\cmidrule(r){5-6}
		\cmidrule(r){7-9}
		& & SFDE & SFET
		& SFDE & SFET
		& SFDE & SFET & NFET\\
		\midrule
		
		\multirow{3}*{$(10,2)$}
		& max & - & 302
		& - & 35.6
		& \multirow{3}*{5\%}
		& \multirow{3}*{0\%}
		& \multirow{3}*{100\%}\\
		& min & 87 & 99
		& 6.1 & 6.5
		& & &\\
		&\setrow{\bfseries} med
		& 127.0 & 128.5
		& 8.9 & 8.6
		& & &\\
		
		\multirow{3}*{$(20,2)$}
		& max & - & -
		& - & -
		& \multirow{3}*{25\%}
		& \multirow{3}*{10\%}
		& \multirow{3}*{100\%}\\
		& min & 115 & 146
		& 33.9 & 40.2
		& & &\\
		&\setrow{\bfseries} med
		& 152.5 & 197.0
		& 49.0 & 58.9
		& & &\\
		
		\multirow{3}*{$(30,2)$}
		& max & - & 366
		& - & 412.0
		& \multirow{3}*{15\%}
		& \multirow{3}*{0\%}
		& \multirow{3}*{100\%}\\
		& min & 86 & 143
		& 72.1 & 103.6
		& & &\\
		&\setrow{\bfseries} med
		& 152.5 & 227.0
		& 124.5 & 170.2
		& & &\\
		
		\multirow{3}*{$(40,2)$}
		& max & - & -
		& - & -
		& \multirow{3}*{15\%}
		& \multirow{3}*{15\%}
		& \multirow{3}*{100\%}\\
		& min & 127 & 190
		& 216.1 & 313.8
		& & &\\
		&\setrow{\bfseries} med
		& 160.0 & 277.0
		& 265.9 & 474.5
		& & &\\
		
		\multirow{3}*{$(50,2)$}
		& max & - & -
		& - & -
		& \multirow{3}*{25\%}
		& \multirow{3}*{20\%}
		& \multirow{3}*{100\%}\\
		& min & 136 & 176
		& 427.2 & 467.2
		& & &\\
		&\setrow{\bfseries} med
		& 200.5 & 290.0
		& 591.4 & 740.2
		& & &\\
		
		\multirow{3}*{$(10,4)$}
		& max & - & 4961
		& - & 1800.1
		& \multirow{3}*{5\%}
		& \multirow{3}*{0\%}
		& \multirow{3}*{100\%}\\
		& min & 102 & 194
		& 32.7 & 59.5
		& & &\\
		&\setrow{\bfseries} med
		& 141.5 & 242.0
		& 43.7 & 83.8
		& & &\\
		
		\multirow{3}*{$(10,6)$}
		& max & - & -
		& - & -
		& \multirow{3}*{10\%}
		& \multirow{3}*{15\%}
		& \multirow{3}*{100\%}\\
		& min & 119 & 206
		& 109.4 & 161.6
		& & &\\
		&\setrow{\bfseries} med
		& 158.5 & 398.5
		& 136.5 & 298.9
		& & &\\
		
		\multirow{3}*{$(10,8)$}
		& max & - & -
		& - & -
		& \multirow{3}*{15\%}
		& \multirow{3}*{10\%}
		& \multirow{3}*{100\%}\\
		& min & 128 & 213
		& 230.7 & 326.5
		& & &\\
		&\setrow{\bfseries} med
		& 174.0 & 327.0
		& 294.9 & 501.1
		& & &\\
		
		\bottomrule
	\end{tabular*}
\end{table}
The results in Tables~\ref{tab:type-1-computational-results}
and~\ref{tab:type-2-computational-results} demonstrate
the computational advantages of
the sequence-form methods over NFET.
For Type-1 games, NFET obtains successful
solutions only for
$(\mathcal{L},\mathcal{A})=(5,2)$
and fails on all 20 instances
in each of the remaining six configurations.
Even for $(5,2)$, where NFET has
slightly fewer variables than SFDE and SFET,
its median computational time is
$340.1$ seconds, compared with
$8.1$ and $8.0$ seconds, respectively.
For Type-2 games, NFET fails
on all instances across all eight
configurations, whereas SFDE and SFET
obtain successful solutions
in every configuration.
These results demonstrate
substantial improvements in computational
efficiency and reliability
achieved by the sequence-form methods.

Although SFDE and SFET share the same
initial point, perturbation,
and theoretical realization-plan path,
they employ different nonlinear
representations of the entropy terms.
These differences may affect
the conditioning of the Jacobian matrices
and the numerical behavior of
the predictor-corrector procedures.
For Type-1 games, SFDE generally requires
fewer iterations and less computational
time than SFET.
The difference is particularly pronounced
for $(8,2)$, where the failure rates
of SFDE and SFET are $45\%$ and $95\%$,
respectively.
For Type-2 games, SFDE achieves
lower median iteration counts
in all eight configurations
and lower median computational times
in seven of them.
However, SFET exhibits lower failure rates
in six configurations,
indicating a trade-off between
computational efficiency and reliability.

Overall, the numerical results demonstrate
the computational benefits of
the sequence-form representation
and show that equivalent entropy
formulations can exhibit different
numerical performance,
despite generating identical
theoretical solution paths.

\section{Conclusion}\label{sec:conclusion}

This paper develops sequence-form formulations
and differentiable path-following methods
for computing the perturbed logistic-QRE variant
in $n$-player extensive-form games
with perfect recall, thereby enabling
the selection of normal-form perfect equilibria.
By constructing an entropy-regularized
artificial game in sequence form,
we establish that its Nash equilibria
characterize the logistic-QRE variant
in the reduced normal form.
Building on this characterization,
we develop a differentiable path-following
method starting from a prescribed
interior realization plan
and establish the existence of a smooth
solution path for almost all sufficiently
small auxiliary perturbations,
whose limiting mixed-strategy profiles
are normal-form perfect.
Furthermore, by reformulating
the dilated-entropy terms as weighted
standard entropy terms,
we derive an equivalent smooth system
that generates the same theoretical
realization-plan path.
The resulting methods also encompass
the computation of standard logistic QRE
as a special case.
Numerical experiments demonstrate
the computational advantages
of the sequence-form methods
over the original normal-form method
and reveal differences in the numerical
performance of the two equivalent formulations.
Future research will focus on improving
numerical robustness for larger games
and extending the proposed approach
to other equilibrium refinements.

	\section*{Funding} This work was partially supported by the Research Grants Council of the Hong Kong SAR Government (Grant No. 11215123) and the National Natural Science Foundation of China (Grant No. 72501218).
	\section*{Disclosure statement} The authors declare that they have no competing interests relevant to this article.
	\section*{Declaration of generative AI use} The authors used OpenAI Codex (GPT-6) to assist with language and LaTeX editing, literature searches, and code review. All AI-assisted suggestions were reviewed by the authors, who take full responsibility for the originality and accuracy of the manuscript, including its mathematical arguments, references, and numerical results.
	
	\bibliographystyle{tfnlm}
\bibliography{library}

@inproceedings{lingWhatGameAre2018,
  title = {What game are we playing? end-to-end learning in normal and extensive form games},
  shorttitle = {What game are we playing?},
  booktitle = {Proc. 27th Int. Jt. Conf. Artif. Intell.},
  author = {Ling, Chun Kai and Fang, Fei and Kolter, J. Zico},
  year = 2018,
  month = jul,
  pages = {396--402},
  publisher = {AAAI Press},
  address = {Stockholm, Sweden}
}

@article{zhanDifferentiablePathfollowingAlgorithm2020,
  title = {A differentiable path-following algorithm for computing perfect stationary points},
  author = {Zhan, Yang and Li, Peixuan and Dang, Chuangyin},
  year = 2020,
  month = jun,
  journal = {Comput. Optim. Appl.},
  volume = {76},
  number = {2},
  pages = {571--588},
  doi = {10.1007/s10589-020-00181-3}
}

@book{hiriarturrutyFundamentalsConvexAnalysis2001,
  author    = {Hiriart-Urruty, Jean-Baptiste and Lemar{\'e}chal, Claude},
  title     = {Fundamentals of Convex Analysis},
  year      = {2001},
  publisher = {Springer-Verlag},
  address   = {Berlin, Heidelberg},
  doi       = {10.1007/978-3-642-56468-0}
}

@book{AllgowerNumericalContinuationMethods1990,
  title = {Numerical Continuation Methods: An Introduction},
  author = {Allgower, Eugene Louis and Georg, Kurt},
  year = 1990,
  volume = {13},
  publisher = {Springer-Verlag},
  address = {Berlin Heidelberg}
}

@article{CaovariantHarsanyitracing2022,
  title = {A variant of Harsanyi's tracing procedures to select a perfect equilibrium in normal form games},
  author = {Cao, Yiyin and Dang, Chuangyin},
  year = 2022,
  month = jul,
  journal = {Games Econ. Behav.},
  volume = {134},
  pages = {127--150},
  doi = {10.1016/j.geb.2022.04.004}
}

@article{caoVariantLogisticQuantal2024,
  title = {A variant of the logistic quantal response equilibrium to select a perfect equilibrium},
  author = {Cao, Yiyin and Chen, Yin and Dang, Chuangyin},
  year = 2024,
  month = jun,
  journal = {J. Optim. Theory Appl.},
  volume = {201},
  number = {3},
  pages = {1026--1062},
  doi = {10.1007/s10957-024-02433-2}
}

@article{Chenextensionquantalresponse2020,
  title = {An extension of quantal response equilibrium and determination of perfect equilibrium},
  author = {Chen, Yin and Dang, Chuangyin},
  year = 2020,
  month = nov,
  journal = {Games Econ. Behav.},
  volume = {124},
  pages = {659--670},
  doi = {10.1016/j.geb.2017.12.023}
}

@article{Chenreformulationbasedsmoothpathfollowing2016,
  title = {A reformulation-based smooth path-following method for computing Nash equilibria},
  author = {Chen, Yin and Dang, Chuangyin},
  year = 2016,
  month = oct,
  journal = {Econ. Theory Bull.},
  volume = {4},
  number = {2},
  pages = {231--246},
  doi = {10.1007/s40505-015-0083-7}
}

@article{Doupnewsimplicialvariable1987,
  title = {A new simplicial variable dimension algorithm to find equilibria on the product space of unit simplices},
  author = {Doup, Timothy Martin and Talman, Albert Jan Johannes},
  year = 1987,
  month = oct,
  journal = {Math. Program.},
  volume = {37},
  number = {3},
  pages = {319--355},
  doi = {10.1007/BF02591741},
  lccn = {3.06}
}

@article{EavesGeneralequilibriummodels1999,
  title = {General equilibrium models and homotopy methods},
  author = {Eaves, Bennett Curtis and Schmedders, Karl},
  year = 1999,
  month = sep,
  journal = {J. Econ. Dyn. Control},
  volume = {23},
  number = {9},
  pages = {1249--1279},
  doi = {10.1016/S0165-1889(98)00073-6}
}

@book{FiaccoIntroductionSensitivityStability1983,
  title = {Introduction to sensitivity and stability analysis in nonlinear programming},
  author = {Fiacco, Anthony Vincent},
  year = 1983,
  month = aug,
  volume = {165},
  publisher = {Academic Press},
  address = {New York},
  doi = {10.1016/S0076-5392(08)X6041-2}
}

@incollection{GarciaSimplicialApproximationEquilibrium1973,
  title = {Simplicial approximation of an equilibrium point for non-cooperative n-person games},
  booktitle = {Mathematical Programming},
  author = {Garcia, Carlos Benito and Lemke, Carlton Edward and Luethi, Heinz},
  editor = {Hu, T. C. and Robinson, Stephen M.},
  year = 1973,
  month = jan,
  pages = {227--260},
  publisher = {Academic Press},
  address = {New York},
  doi = {10.1016/B978-0-12-358350-5.50011-7}
}

@article{GovindanglobalNewtonmethod2003,
  title = {A global Newton method to compute Nash equilibria},
  author = {Govindan, Srihari and Wilson, Robert},
  year = 2003,
  month = may,
  journal = {J. Econ. Theory},
  volume = {110},
  number = {1},
  pages = {65--86},
  doi = {10.1016/S0022-0531(03)00005-X}
}

@article{govindanStructureTheoremsGame2002,
  title = {Structure theorems for game trees},
  author = {Govindan, Srihari and Wilson, Robert},
  year = 2002,
  month = jun,
  journal = {Proc. Natl. Acad. Sci.},
  volume = {99},
  number = {13},
  pages = {9077--9080},
  publisher = {Proceedings of the National Academy of Sciences},
  doi = {10.1073/pnas.082249599}
}

@article{HeringsComputationNashEquilibrium2002,
  title = {Computation of the Nash equilibrium selected by the tracing procedure in n-person games},
  author = {Herings, Pierre Jean-Jacques and {van den Elzen}, Antoon},
  year = 2002,
  month = jan,
  journal = {Games Econ. Behav.},
  volume = {38},
  number = {1},
  pages = {89--117},
  doi = {10.1006/game.2001.0856}
}

@article{Heringsdifferentiablehomotopycompute2001,
  title = {A differentiable homotopy to compute Nash equilibria of n-person games},
  author = {Herings, Pierre Jean-Jacques and Peeters, Ronald Johannes Antonius Petrus},
  year = 2001,
  month = jul,
  journal = {Econ. Theory},
  volume = {18},
  number = {1},
  pages = {159--185},
  doi = {10.1007/PL00004129}
}

@misc{houCharacterizationComputationNormalForm2026a,
  title = {Characterization and computation of normal-form proper equilibria in extensive-form games via the sequence-form representation},
  author = {Hou, Yuqing and Cao, Yiyin and Dang, Chuangyin},
  howpublished = {{arXiv} preprint {arXiv}:2602.10524},
  year = 2026,
  month = feb,
  number = {arXiv:2602.10524},
  eprint = {2602.10524},
  primaryclass = {cs},
  publisher = {arXiv},
  doi = {10.48550/arXiv.2602.10524},
  archiveprefix = {arXiv}
}

@misc{houSequenceFormCharacterizationDifferentiable2025,
  title = {A sequence-form characterization and differentiable path-following method for computing normal-form perfect equilibria in extensive-form games},
  author = {Hou, Yuqing and Cao, Yiyin and Dang, Chuangyin and Wang, Yong},
  howpublished = {{arXiv} preprint {arXiv}:2505.13827},
  year = 2025,
  month = nov,
  number = {arXiv:2505.13827},
  eprint = {2505.13827},
  primaryclass = {cs},
  publisher = {arXiv},
  doi = {10.48550/arXiv.2505.13827},
  archiveprefix = {arXiv}
}

@article{houSequenceformDifferentiablePathfollowing2025,
  title = {A sequence-form differentiable path-following method to compute Nash equilibria},
  author = {Hou, Yuqing and Cao, Yiyin and Dang, Chuangyin and Wang, Yong},
  year = 2025,
  month = sep,
  journal = {Comput. Optim. Appl.},
  volume = {92},
  number = {1},
  pages = {265--300},
  doi = {10.1007/s10589-025-00702-y}
}

@article{Kollercomplexitytwopersonzerosum1992,
  title = {The complexity of two-person zero-sum games in extensive form},
  author = {Koller, Daphne and Megiddo, Nimrod},
  year = 1992,
  month = oct,
  journal = {Games Econ. Behav.},
  volume = {4},
  number = {4},
  pages = {528--552},
  doi = {10.1016/0899-8256(92)90035-Q}
}

@article{KollerEfficientComputationEquilibria1996,
  title = {Efficient computation of equilibria for extensive two-person games},
  author = {Koller, Daphne and Megiddo, Nimrod and {von Stengel}, Bernhard},
  year = 1996,
  month = jun,
  journal = {Games Econ. Behav.},
  volume = {14},
  number = {2},
  pages = {247--259},
  doi = {10.1006/game.1996.0051}
}

@article{KollerRepresentationssolutionsgametheoretic1997,
  title = {Representations and solutions for game-theoretic problems},
  author = {Koller, Daphne and Pfeffer, Avi},
  year = 1997,
  month = jul,
  journal = {Artif. Intell.},
  volume = {94},
  number = {1},
  pages = {167--215},
  doi = {10.1016/S0004-3702(97)00023-4},
  lccn = {14.05}
}

@article{KuhnExtensiveGames1950,
  title = {Extensive games},
  author = {Kuhn, Harold William},
  year = 1950,
  month = oct,
  journal = {Proc. Natl. Acad. Sci.},
  volume = {36},
  number = {10},
  pages = {570--576},
  publisher = {National Academy of Sciences},
  doi = {10.1073/pnas.36.10.570},
  lccn = {1}
}

@article{liArbitraryStartingTracing2020,
  title = {An Arbitrary Starting Tracing Procedure for Computing Subgame Perfect Equilibria},
  author = {Li, Peixuan and Dang, Chuangyin},
  year = 2020,
  month = aug,
  journal = {J. Optim. Theory Appl.},
  volume = {186},
  number = {2},
  pages = {667--687},
  doi = {10.1007/s10957-020-01703-z},
  lccn = {3}
}

@article{mas-colellNoteTheoremBrowder1974,
  title = {A note on a theorem of F. Browder},
  author = {{Mas-Colell}, Andreu},
  year = 1974,
  month = dec,
  journal = {Math. Program.},
  volume = {6},
  number = {1},
  pages = {229--233},
  doi = {10.1007/BF01580239},
  lccn = {2}
}

@article{mckelveyQuantalResponseEquilibria1995,
  title = {Quantal Response Equilibria for Normal Form Games},
  author = {McKelvey, Richard D. and Palfrey, Thomas R.},
  year = 1995,
  month = jul,
  journal = {Games Econ. Behav.},
  volume = {10},
  number = {1},
  pages = {6--38},
  doi = {10.1006/game.1995.1023}
}

@article{MyersonRefinementsNashequilibrium1978,
  title = {Refinements of the Nash equilibrium concept},
  author = {Myerson, Roger Bruce},
  year = 1978,
  month = jun,
  journal = {Int. J. Game Theory},
  volume = {7},
  number = {2},
  pages = {73--80},
  doi = {10.1007/BF01753236},
  lccn = {0.649}
}

@article{NashEquilibriumpointsnperson1950,
  title = {Equilibrium points in n-person games},
  author = {Nash, John F.},
  year = 1950,
  month = jan,
  journal = {Proc. Natl. Acad. Sci.},
  volume = {36},
  number = {1},
  pages = {48--49},
  publisher = {Proceedings of the National Academy of Sciences},
  doi = {10.1073/pnas.36.1.48}
}

@book{OsborneCourseGameTheory1994,
  title = {A Course in Game Theory},
  author = {Osborne, Martin John and Rubinstein, Ariel},
  year = 1994,
  month = jul,
  volume = {1},
  publisher = {The MIT Press},
  address = {Cambridge},
  googlebooks = {mnv1DAAAQBAJ}
}

@article{romanovskiiReductionGameComplete1962,
  title = {Reduction of a game with complete memory to a matrix game},
  author = {Romanovskii, I. V.},
  year = 1962,
  journal = {Dokl.akad.nauk Sssr},
  volume = {3},
  number = {3},
  pages = {62--64}
}

@article{RosenmullerGeneralizationLemkeHowson1971,
  title = {On a generalization of the Lemke--Howson algorithm to noncooperative n-person games},
  author = {Rosenm{\"u}ller, Johann},
  year = 1971,
  month = jul,
  journal = {SIAM J. Appl. Math.},
  volume = {21},
  number = {1},
  pages = {73--79},
  publisher = {{Society for Industrial and Applied Mathematics}},
  doi = {10.1137/0121010},
  lccn = {2.148}
}

@article{SeltenReexaminationperfectnessconcept1975,
  title = {Reexamination of the perfectness concept for equilibrium points in extensive games},
  author = {Selten, Reinhard},
  year = 1975,
  month = mar,
  journal = {Int. J. Game Theory},
  volume = {4},
  number = {1},
  pages = {25--55},
  doi = {10.1007/BF01766400},
  lccn = {4}
}

@article{turocyDynamicHomotopyInterpretation2005,
  title = {A dynamic homotopy interpretation of the logistic quantal response equilibrium correspondence},
  author = {Turocy, Theodore L.},
  year = 2005,
  month = may,
  journal = {Games Econ. Behav.},
  volume = {51},
  number = {2},
  pages = {243--263},
  doi = {10.1016/j.geb.2004.04.003}
}

@article{vandenelzenProcedureFindingNash1991,
  title = {A procedure for finding Nash equilibria in bi-matrix games},
  author = {{van den Elzen}, Arnoldus Hendrikus and Talman, Albert Jan Johannes},
  year = 1991,
  month = jan,
  journal = {Z. Oper. Res.},
  volume = {35},
  number = {1},
  pages = {27--43},
  doi = {10.1007/BF01415958}
}

@article{vanderLaanComputationFixedPoints1982,
  title = {On the computation of fixed points in the product space of unit simplices and an application to noncooperative N person games},
  author = {{van der Laan}, Gerrit and Talman, Albert Jan Johannes},
  year = 1982,
  month = feb,
  journal = {Math. Oper. Res.},
  volume = {7},
  number = {1},
  pages = {1--13},
  doi = {10.1287/moor.7.1.1},
  lccn = {2.215}
}

@article{vonStengelComputingNormalForm2002,
  title = {Computing normal form perfect equilibria for extensive two-person games},
  author = {{von Stengel}, Bernhard and {van den Elzen}, Antoon and Talman, Dolf},
  year = 2002,
  month = mar,
  journal = {Econometrica},
  volume = {70},
  number = {2},
  pages = {693--715},
  doi = {10.1111/1468-0262.00300},
  lccn = {1}
}

@article{vonStengelEfficientComputationBehavior1996,
  title = {Efficient computation of behavior strategies},
  author = {{von Stengel}, Bernhard},
  year = 1996,
  month = jun,
  journal = {Games Econ. Behav.},
  volume = {14},
  number = {2},
  pages = {220--246},
  doi = {10.1006/game.1996.0050}
}

@article{WilsonComputingEquilibriaNPerson1971,
  title = {Computing equilibria of n-person games},
  author = {Wilson, Robert},
  year = 1971,
  month = jul,
  journal = {SIAM J. Appl. Math.},
  volume = {21},
  number = {1},
  pages = {80--87},
  publisher = {{Society for Industrial and Applied Mathematics}},
  doi = {10.1137/0121011},
  lccn = {2.148}
}

@article{marchiRefinementProper1995,
  title = {Refinement of proper equilibrium points},
  author = {Marchi, E. and Neme, A.},
  year = 1995,
  journal = {Optimization},
  volume = {32},
  number = {2},
  pages = {185--192},
  doi = {10.1080/02331939508844046}
}

@article{drevesBestResponse2019,
  title = {A best-response approach for equilibrium selection in two-player generalized Nash equilibrium problems},
  author = {Dreves, Axel},
  year = 2019,
  journal = {Optimization},
  volume = {68},
  number = {12},
  pages = {2269--2295},
  doi = {10.1080/02331934.2019.1646743}
}

@article{liuDifferentiableRobust2025,
  title = {A differentiable path-following method to compute Nash equilibria in robust normal-form games},
  author = {Liu, Lingyuan and Chen, Yin and Dang, Chuangyin},
  year = 2025,
  journal = {Optimization},
  volume = {74},
  number = {3},
  pages = {549--588},
  doi = {10.1080/02331934.2023.2259925}
}

\appendix
\section{The Boundedness of $\widetilde{\mathscr{S}}_D$ in Theorem~\ref{thm:connected-solution-component}}\label{app:multiplier-boundedness}
This appendix establishes the boundedness
of $\widetilde{\mathscr{S}}_D$,
as required in the proof of
Theorem~\ref{thm:connected-solution-component}. 

For any sequence $\varpi^i\in W^i$, we use $M_i^+(\varpi^i)$ to denote the index set of all information sets of player $i$ that may arise after $\varpi^i$, not necessarily immediately. Let $(\bm{\gamma}^*,\bm{\nu}^*,t)\in\widetilde{\mathscr{S}}_D$ be a solution of System~\eqref{eq:sequence-form-equilibrium-system}. Applying backward induction to the first group of equations in System~\eqref{eq:sequence-form-equilibrium-system}, we obtain, for each $i\in N$ and $j\in M_i$,
\begin{equation}\label{eq:multiplier-recursion}
	\begin{array}{l}
		-\nu^{*i}_{I^j_i}-1+\sum\limits_{j_q\in M^+_i(\varpi^i_{I^j_i}),a_q\in A(I^{j_q}_i)}\frac{\gamma^{*i}(\varpi^i_{I^{j_q}_i}a_q)}{\gamma^{*i}(\varpi^i_{I^{j}_i})}\biggl((1-c(t))g^i(\varpi^i_{I^{j_q}_i}a_q,\bm{\pi_w}(\bm{\gamma}^{*-i},t))\\
		\hspace{5.5cm}-c(t)\Bigl( \ln \frac{\gamma^{*i}(\varpi^i_{I^{j_q}_i}a_q)}{\gamma^{*i}(\varpi^i_{I^{j_q}_i})} -\ln\frac{\gamma^{0i}(\varpi^i_{I^{j_q}_i}a_q)}{\gamma^{0i}(\varpi^i_{I^{j_q}_i})}\Bigr)\biggr)=0.
	\end{array}
\end{equation}
To see this, consider $i\in N$,  $j\in M_i$ such that $m_i(\varpi^i_{I_i^j}a)=0$ for every $a\in A(I_i^j)$, \eqref{eq:multiplier-recursion} follows directly from the first group of equations in System~\eqref{eq:sequence-form-equilibrium-system} after multiplying them by $\gamma^{*i}(\varpi^i_{I^{j}_i}a)/\gamma^{*i}(\varpi^i_{I^{j}_i})$ and summing over $a\in A(I_i^j)$. Now consider the case in which $m_i(\varpi^i_{I_i^j}a)>0$ for some
$a\in A(I_i^j)$. Suppose that \eqref{eq:multiplier-recursion} has already been established for every $j_l\in M_i(\varpi^i_{I_i^j}a)$. Substituting the corresponding identity for $\zeta^i_{I_i^j}(a)$ into the first group of equations in System~\eqref{eq:sequence-form-equilibrium-system} gives the desired expression for each $a\in A(I_i^j)$. Multiplying this expression by $\gamma^{*i}(\varpi^i_{I^{j}_i}a)/\gamma^{*i}(\varpi^i_{I^{j}_i})$ and summing over $a\in A(I_i^j)$ gives \eqref{eq:multiplier-recursion} for $i\in N$,  $j\in M_i$.

We know that $-e^{-1}\leq v\ln v\leq 0$ for $0<v\leq 1$. Let $U^i_l=\min_{h\in Z}u^i(h)$, $U^i_u=\max_{h\in Z}u^i(h)$, and $Y^i_l=\min_{\varpi^i\in W^i}\gamma^{0i}(\varpi^i)$. It then follows from \eqref{eq:multiplier-recursion} that, for any $i\in N$ and 
$j\in M_i$, \[|W^i|(-|U^i_l|+\ln Y^i_l)-1\leq\nu^{*i}_{I^j_i}\leq |W^i|(|U^i_u|+e^{-1})-1.\]

\section{Jacobian Matrix of $p(\cdot)$ Full-Row Rank Proof in Theorem~\ref{thm:smooth-perfect-equilibrium-path}}\label{app:jacobian-full-row-rank}

This appendix proves that the Jacobian matrix $Dp(\bm{\gamma},\bm{\nu},t;\bm{\alpha})$ of $p(\bm{\gamma},\bm{\nu},t;\bm{\alpha})$ has full row rank on $\text{int}(\Lambda)\times\mathbb{R}^{m_0}\times(0,1)\times\mathbb{R}^{n_0}$, which is critical for the proof of Theorem~\ref{thm:smooth-perfect-equilibrium-path}.

Let $g(\bm{\gamma},\bm{\nu},t;\bm{\alpha})$ denote
the first $n_0$ components of
$p(\bm{\gamma},\bm{\nu},t;\bm{\alpha})$.
For $\bm{\gamma}\in\Lambda_{++}$ and
$t\in(0,1]$, the Jacobian matrix is
\begin{equation*}
	Dp(\bm{\gamma},\bm{\nu},t;\bm{\alpha})
	=
	\begin{pmatrix}
		D_{\bm{\gamma}} g
		&-\bm B^\top
		&D_tg
		&-c(t)(1-t)\bm I^{n_0\times n_0}\\
		\bm B&0&0&0
	\end{pmatrix},
\end{equation*}
where $\bm I^{n_0\times n_0}$ is
the $n_0\times n_0$ identity matrix
and $\bm B=\bm{\bar{B}}-\bm{\tilde{B}}$.
The matrix $\bm{\bar{B}}$ is given by
\[
\bm{\bar{B}}
=
\operatorname{diag}
\left(
(\bm e_i^j)^\top:
i\in N,\ j\in M_i
\right),
\]
where
$\bm e_i^j=(1,\ldots,1)^\top
\in\mathbb R^{|A(I_i^j)|}$.
The matrix $\bm{\tilde{B}}\in \mathbb{R}^{m_0\times n_0}$ is defined such that, in each row, the element corresponding to the sequence associated with the relevant information set takes the value $1$, whereas all remaining elements are set to $0$. We observe that ${\bm I}^{n_0\times n_0}$ and ${\bm B}$ are of full-row rank. Thus, for any $t\in(0,1)$, the Jacobian matrix $Dp(\bm{\gamma},\bm{\nu},t;\bm{\alpha})$ is of full-row rank.

When $t=1$, the auxiliary perturbation
vanishes, and the Jacobian matrix
takes the block form
\[
Dp(\bm{\gamma},\bm{\nu},1;\bm{\alpha})
=
\begin{pmatrix}
	\bm G & -\bm B^\top\\
	\bm B & 0
\end{pmatrix},
\]
where
$\bm G=D_{\bm{\gamma}} g(\bm{\gamma},\bm{\nu},1;\bm{\alpha})$.
To establish its nonsingularity,
suppose that
\[
\begin{pmatrix}
	\bm G & -\bm B^\top\\
	\bm B & 0
\end{pmatrix}
\begin{pmatrix}
	\bm{x}\\\bm{y}
\end{pmatrix}
=0.
\]
Then $\bm G\bm{x}-\bm B^\top \bm{y}=0,
\bm B\bm{x}=0$,
which implies
$\bm{x}^\top\bm G\bm{x}
=\bm{x}^\top\bm B^\top \bm{y}=0$.
For any $\bm{x}\in\ker\bm B$,
the realization-plan constraints imply
\[
\sum_{a\in A(I_i^j)}
x^i(\varpi^i_{I_i^j}a)
=
x^i(\varpi^i_{I_i^j}),
\]
where $x^i(\emptyset)=0$.
Differentiating the logarithmic terms
in System~\eqref{eq:sequence-form-equilibrium-system}
at $t=1$ and using these constraints,
we obtain
\[
	\bm{x}^\top\bm G\bm{x}
	=-
	\sum_{i\in N}
	\sum_{j\in M_i}
	\sum_{a\in A(I_i^j)}
	\gamma^i(\varpi^i_{I_i^j}a)
	\left(
	\frac{x^i(\varpi^i_{I_i^j}a)}
	{\gamma^i(\varpi^i_{I_i^j}a)}
	-
	\frac{x^i(\varpi^i_{I_i^j})}
	{\gamma^i(\varpi^i_{I_i^j})}
	\right)^2.
\]
Since $\bm{\gamma}\in\Lambda_{++}$,
the right-hand side is nonpositive,
with equality only if $\bm{x}=0$.
Indeed, equality requires
\[
\frac{x^i(\varpi^i_{I_i^j}a)}
{\gamma^i(\varpi^i_{I_i^j}a)}
=
\frac{x^i(\varpi^i_{I_i^j})}
{\gamma^i(\varpi^i_{I_i^j})}
\]
for every $i\in N$, $j\in M_i$,
and $a\in A(I_i^j)$.
Together with $x^i(\emptyset)=0$,
this implies $\bm{x}=0$ by induction
on the action sequences. Consequently, $\bm B^\top \bm{y}=0$.
Since $\bm B$ has full row rank,
we obtain $\bm{y}=0$.
Thus,
$Dp(\bm{\gamma},\bm{\nu},1;\bm{\alpha})$
is nonsingular, which implies that
$Dp(\bm{\gamma},\bm{\nu},1;\bm{\alpha})$
has full row rank.

Combining the two cases,
we conclude that
$Dp(\bm{\gamma},\bm{\nu},t;\bm{\alpha})$
has full row rank for all
$\bm{\gamma}\in\Lambda_{++}$,
$\bm{\nu}\in\mathbb R^{m_0}$,
$t\in(0,1]$,
and $\bm{\alpha}\in\mathbb R^{n_0}$.

\end{document}